\documentclass{statsoc}

\usepackage[letterpaper,left=2.2cm, right=2.2cm, top=2.2cm, bottom=2.2cm]{geometry}
\usepackage[utf8]{inputenc}

\RequirePackage{amsmath,amsfonts,amssymb}
\RequirePackage[authoryear]{natbib}
\RequirePackage[colorlinks,citecolor=blue,urlcolor=blue]{hyperref}
\RequirePackage{graphicx}
\usepackage{dsfont}
\usepackage{bigints}
\usepackage{tikz}
\usepackage{filecontents}
\usepackage{multirow}
\RequirePackage{bm}
\usepackage{xfrac}
\newcommand{\overbar}[1]{\mkern 1.5mu\overline{\mkern-1.5mu#1\mkern-1.5mu}\mkern 1.5mu}

\newtheorem{definition}{Definition}
\newtheorem{proposition}{Proposition}
\newcommand{\gtheta}{g_{_\theta}\hspace{-0.05cm}}
\newcommand{\gthetahat}{g_{_{\widehat{\theta}}}\hspace{-0.02cm}}
\newcommand{\btheta}{b_{_\theta}\hspace{-0.05cm}}
\newcommand{\gthetaxy}{\gtheta(x,z)}

\newcommand{\Etilde}{\widetilde{E}_{_{\theta,T}}}
\newcommand{\Etheta}{E_{_\theta}\hspace{-0.05cm}}
\newcommand{\psitheta}{\psi_{_\theta}\hspace{-0.05cm}}
\newcommand{\phithetaj}{r_{j}\hspace{-0.02cm}}
\newcommand{\phithetatildej}{\widetilde{r}_{j}\hspace{-0.02cm}}
\newcommand{\psithetatildej}{s_{j}\hspace{-0.02cm}}
\newcommand{\bxy}{\btheta(x,z)}
\newcommand{\gthetak}{\gtheta \bigl(x_k,\nut\bigl)}
\newcommand{\klist}{k=1,\dots,K}

\newcommand{\mdotx}{\dot{m}_{_{\theta}}\hspace{-0.05cm}(x)}
\newcommand{\ellthetat}{\ell_{t}\hspace{-0.02cm}}

\newcommand{\mtheta}{m_{_{\theta}}\hspace{-0.05cm}(x_k)}
\newcommand{\mthetax}{m_{_{\theta}}\hspace{-0.05cm}(x)}
\newcommand{\mtildex}{\widetilde{m}_{_{\theta,T}}\hspace{-0.05cm}(x)}
\newcommand{\mtildexk}{\widetilde{m}_{_{\theta,T}}\hspace{-0.05cm}(x_k)}

\newcommand{\mthetahat}{m_{_{\widehat{\theta}}}\hspace{-0.05cm}(x_k)}
\newcommand{\mutheta}{\mu_{_{\theta\hspace{-0.03cm},\hspace{-0.04cm}K}}\hspace{-0.05cm}}
\newcommand{\muk}{\mu_{_K}\hspace{-0.05cm}}
\newcommand{\NT}{N_{_{T}}{\hspace{-0.05cm}}}
\newcommand{\observed}{\nut}
\newcommand{\parallela}{\scalebox{0.5}{$\slash\!\slash$}}
\newcommand{\gparallel}{\gtheta^{\parallela}}
\newcommand{\gparallelhat}{\gthetahat^{\parallela}}
\newcommand{\sumk}{\sum_{k=1}^K}
\newcommand{\thetahat}{\widehat{\theta}}
\newcommand{\Real}{\mathbb{R}}
\newcommand{\NN}{\mathbb{N}}
\newcommand{\Vtilde}{\widetilde{V}_{_{\theta,T}}}
\newcommand{\vtheta}{v_{_{\theta\hspace{-0.03cm},\hspace{-0.04cm}K}}\hspace{-0.05cm}}

\newcommand{\vthetahat}{v_{_{\widehat{\theta}\hspace{-0.03cm},\hspace{-0.04cm}K}}\hspace{-0.05cm}}
\newcommand{\X}{\mathcal{X}}
\def\bea{\begin{eqnarray}}
\def\eea{\end{eqnarray}}
\def\beas{\begin{eqnarray*}}
\def\eeas{\end{eqnarray*}}
\def\Deltax{\Delta[x_k]}

\newcommand{\nut}{\nu(x_k)}
\newcommand{\nutx}{\nu(x)}

\newcommand{\wh}{\widehat}

\definecolor{darkred2}{HTML}{8B0000}

\newcommand{\sara}[1]{{\color{black}#1}}

\title[On the statistical analysis of grouped data]{On the statistical analysis of grouped data: when Pearson $\chi^2$ and other divisible statistics are not goodness-of-fit}
\author[Algeri S. and Khmaladze E.V.]{Sara Algeri$^1$ and Estate V. Khmaladze$^2$}
\address{$^1$ School of Statistics, University of Minnesota,  Minneapolis, MN, USA.\\ Email: salgeri@umn.edu}
\address{$^2$ School of Mathematics and Statistics, Victoria University of Wellington, Wellington, New Zealand. \\Email: estate.khmaladze@vuw.ac.nz}

\begin{document}

\begin{abstract}
Thousands of experiments are analyzed, and papers are published each year involving the statistical analysis of grouped data. While this area of statistics is often perceived – somewhat naively – as saturated, several misconceptions still affect everyday practice, and new frontiers have so far remained unexplored. Researchers must be aware of the limitations affecting their analyses and what new possibilities are at their hands.\\
The article introduces a unifying approach to the analysis of divisible statistics – that includes Pearson’s $\chi^2$, the likelihood ratio and spectral statistics, as special cases –  when a statistician deals with a large number of bins/groups, thus leading to a large number of small or moderate frequencies. Performance of the tests is analyzed against the class of contiguous (local) alternatives. \\
Perhaps the most surprising result here is that, in this `sparse' regime, most of the tests proposed in the literature can be modified to produce more powerful tests, and no single test based on a divisible statistic leads to a goodness-of-fit test. Distribution-free goodness-of-fit tests are also constructed.
\end{abstract}

\section{Introduction}
\label{sec:introduction}
The rise of goodness-of-fit in statistics can be traced back to a correspondence between Edgeworth and Pearson at the end of the 19th century  \citep[cf.][]{stigler} on testing regarding the multinomial distribution of the 37 numbers on the roulette wheel, which led to Pearson's seminal paper on the $\chi^2$-test \citep{pearson1900}. The latter can be considered one of the pillars of data analysis and, similarly to other foundational tools, it is still dominant in everyday practice.

Given $K$ mutually exclusive groups or cells, labeled by $x_k$, let $\nut$ be the frequencies observed in each group.
Pearson's $\chi^2$ aims to test the hypothesis
\begin{equation}
    \label{eqn:test}
   H_0:E[\nut] =\mtheta \quad \text{ for all $k\in\{1,\dots,K\}$}
\end{equation}
where the expected frequencies, $\mtheta$, reflect our understanding of the phenomenon under study and may depend on a  --  possibly unknown  --  parameter $\theta\in \Real^p$. The  prescribed test statistic is
\[\sumk\frac{\bigl(\nut-\mtheta\bigl)^2}{\mtheta},\]
and it is  known to be $\chi^2$ distributed under $H_0$ if  $\min_{x_k}\{\mtheta\}\rightarrow \infty$.

Following Pearson's result, many different statistics have been introduced in the literature as substitutes to the $\chi^2$-test or supplementary to it \citep[cf.,][]{cochran}.
They can be classified into two main kinds: those requiring the data to be in a grouped form and those that apply to continuous data. 

\emph{\textbf{Theoretical unification in the analysis of continuous data.}} An important feature of many goodness-of-fit statistics for continuous data is that they can be specified as functionals of the \emph{empirical process}. 
Letting $y_1,\dots,y_n$ be i.i.d. realizations of a real-valued random variable $Y$, distributed according to the distribution $F$,  the empirical process  $w_n(y)$ is such that  
\[w_n(y)=\sqrt{n}[F_n(y)-F(y)],\]
with $F_n(y)=\frac{1}{n}\sum_{i=1}^n\mathds{1}_{\{y_i\leq y\}}$ being the empirical distribution function.  A prominent example of a goodness-of-fit test based on $w_n(y)$ is the Kolmogorov-Smirnov statistic \citep[][]{kolmogorov,smirnov} 
\sara{
\[\sup_y |w_n|.\]}
Also very prominent are the Cramer-von Mises \citep[][]{smirnov37} and Anderson-Darling \citep[][]{anderson} statistics given by
\[\int w^2_n(y)  dF(y) \qquad\text{and} \qquad\bigintsss\hspace{-0.1cm} \frac{w^2_n(y)}{F(y)(1-F(y))} dF(y),
\]
respectively. Several other constructs used in the analysis of continuous data can be expressed as linear functionals of $w_n(y)$  \citep[e.g.,][]{wellnerrev}.  Such unification has led to notable advancements in different areas of statistics and has enabled developments in the theory and practice of goodness-of-fit, which continue at present
\citep[e.g.,][]{henze,khm16,moscovich, durio,algeri22,dumbgen,roberts}. 

\emph{\textbf{An attempt of unification in the analysis of grouped data.}} Alternatives to Pearson in the analysis of grouped data include the likelihood ratio
\[\sumk\nut\ln\frac{\mtheta}{\nut},\]
the linear statistic
\[\sumk\bigl[\nut-\mtheta\bigl],\]
and the (cumulative) spectral statistic
\[\sumk\mathds{1}_{\{\observed\leq q\}}\quad \text{with $q \in \NN$,}\]
which plays a central role in occupancy and diversity problems \citep[e.g.,][]{magurran,mirakhmedov,barbour,khm11}.
The commonality of these and many other tests is that they can all be expressed as the sum of a function of the observed and expected frequencies. They have been referred to in the literature
as `divisible statistics' (or sometimes `separable statistics') \citep[cf.][]{medvedev70}. 

Given a function $g(z,t)$, usually differentiable in $t$, a divisible statistic based on such a function has the form
\begin{equation}
\label{eqn:classical_ds}
    \frac{1}{\sqrt{K}}\sumk g\bigl(\nut,\mtheta\bigl).
    \end{equation}
The study of the class of divisible statistics as such began in the last quarter of the 20th century
\citep[e.g.,][]{ivchenko79,ivchenko81, khm84, mnatsakanov86,mnatsakanov88}.
Of particular interest is their behavior in the high-dimensional regime, also explored in recent literature on minimax optimal tests \citep[e.g.,][]{balakrishnan2,balakrishnan1,cai}, in which the number of cells is not fixed and grows with the size of the sample.

Formally, letting $T=\sumk \mtheta$ be the expected sample size, we assume that $K$ and $T$ increase simultaneously so that
\begin{equation}
\label{eqn:A2}
\frac{T}{K} \to c \in (0,\infty),\quad\text{ as $T\rightarrow\infty$ and $K\rightarrow\infty$,}
\end{equation}
and $c$ can be interpreted as the average number of expected events per cell. 
Note that, differently from the case in which $T$ grows at a faster rate than $K$,  by assuming  $T$ and $K$ have the same order of magnitude,  we allow the frequencies to be potentially small and, thus, not necessarily Gaussian in the limit. 
The condition in \eqref{eqn:A2} is necessary 
 to describe many commonly encountered phenomena in which replicate observations are available, but only in a limited number. For example, 
 in Edgeworth and Pearson's epistolary communication, the roulette test relied on the normal approximation of the frequencies.  This is sensible since the underlying distribution has 36 dimensions. But if we were to ask the same question today considering, instead of a roulette wheel, a poker machine with about 20 images on 15 wheels, this would give us $20^{15}\sim 10^{19}$ dimensions or cells, which could not be easily filled in any real experiment.  
 In the statistical analysis of large texts, the asymptotic described by \eqref{eqn:A2}  is practically unavoidable: in a text composed of several hundred thousand words, more than $90\%$  of the words that constitute its vocabulary are used no more than $10$ times each, and about $50\%$ of the words are used only once \citep[e.g.,][]{baayen}. In sociological surveys, the diversity of opinions is studied through questionnaires, and the situation is similar: in a questionnaire of, say, $15$ binary questions, the number of different possible opinions, or groups, is $2^{15}\sim 33,000$. In a sample of $100,000$ respondents, about $78\%$ of different opinions will be expressed no more than $10$ times; in balanced questionnaires, the probability of the answer `yes' is close to $1/2$ and $c$ will be very close to $3$ or $4$ \citep[e.g.,][]{khm11}. The asymptotics in \eqref{eqn:A2} are also typical in occupation problems in environmental ecology. In this field, most models involve a large number of different species populating a certain region. When an observational experiment is conducted over a given period of time, the number of different species observed is also large, but individuals from the same species may only be observed a handful of times \citep[e.g.,][]{magurran}. 
As discussed in detail below, the assumption in \eqref{eqn:A2} is also relevant in counting experiments in physics and astronomy. 
In these and many other instances, researchers have so far been forced to group their data in large groups to use classical asymptotic theory in which $c$ increases. This manuscript demonstrates that this grouping is not necessary. An asymptotic theory under \eqref{eqn:A2} very much exists, although it is not the same theory as in the case in which the frequencies are accumulated.

\sara{The existing literature on the class of divisible statistics focuses on goodness-of-fit for simple hypotheses. One of the most notable results under \eqref{eqn:A2} is that, if $\theta$ is known, no single divisible statistic can detect all local departures from the null, somewhat diminishing their adequacy for goodness-of-fit (as per Definition \ref{def:gof_propetry} in Section \ref{sec:alternatives}). However, such alternatives can be detected with non-zero limiting power when combining many divisible statistics, e.g., by considering functionals of partial sums \citep[cf.][]{khm84} (see also Section \ref{sec:recovering_GOF}).
}  The effect of parameter estimation, on the other hand, has only been investigated on a case-by-case basis \citep[cf.,][]{cressie} or for sub-classes such as that of power divergence statistics \citep{cressie_read,muller03}.  
No parametric result on the class of divisible statistics as a whole has yet been introduced, nor has an analog to the empirical process been proposed for the analysis of grouped data.

\emph{\textbf{Divisible statistics anew: statistical motivation and relevance to the physical sciences.}}  The primary theoretical motivation of this manuscript is the need to unify the theory of statistical inference for grouped data to an extent comparable to what is enabled by the theory of empirical processes in the continuous regime. 
In Section \ref{sec:divisible}, we show that it is possible to express different statistics involved in the analysis of grouped data as linear functionals based on the same random measure. This construction not only brings the desired unification but also unveils new phenomena \sara{that have no analog} in the analysis of i.i.d. observations. For example, a somewhat unanticipated result is presented in Section \ref{sec:linear}:  in a sparse regime, all tests based on divisible statistics are dominated by weighted linear statistics. 

\begin{figure}[!h]
    \centering
        \includegraphics[width=130mm]{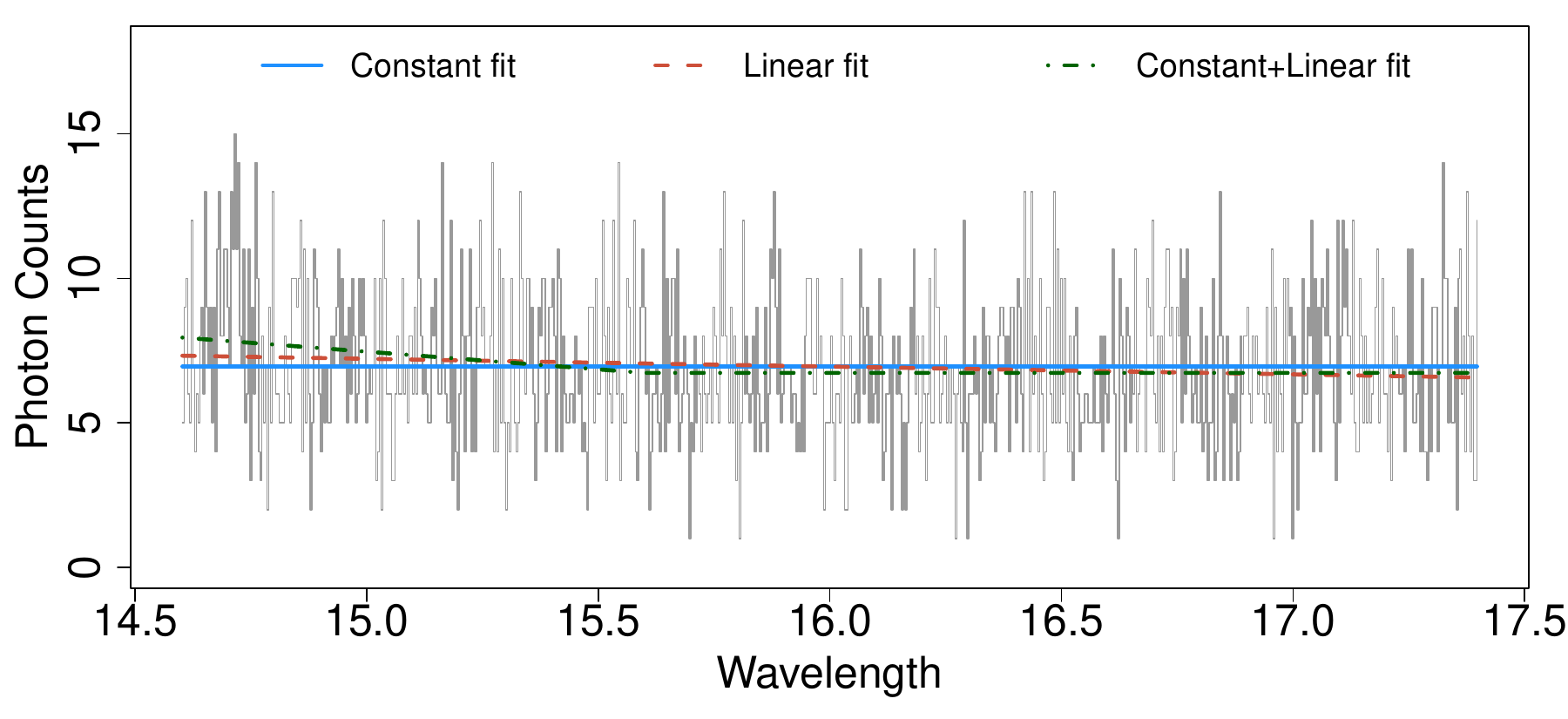}     
    \caption{$X$-ray source-free spectrum from a Chandra observation and discretized in $750$ bins. The blue solid, red dashed, and green chained lines are, respectively, the best fits obtained for constant, linear, and piecewise linear mean functions. 
    }
    \label{fig:RTCru_spectrum}
\end{figure} 
In what follows, we assume the observed frequencies are independent Poisson random variables and  \eqref{eqn:A2} holds  --  that is, asymptotically, the frequencies remain Poisson and do not reach the Gaussian limit. We also assume that, given a compact subset of $\X$  of $\Real^d$ with Lebesgue measure $|\X|$, the grouping leads to an exhaustive partition of $\X$ into bins.

These assumptions are relevant, for example, in physics and astronomy, where the detector mechanism often leads to a discretization of the area under study. As a result, the data consists of event counts observed in the bins  --  be it channels, pixels, voxels, or other segments  --  of the discretized space. \sara{Even} when the detector resolution is so high that it leads to a nearly continuous data stream,   the statistical analysis is often conducted by partitioning the data into bins. This choice may be dictated by the need for fast computations and numerical stability when fitting the models under study, to improve the signal-to-noise ratio, and/or to simplify the model assessment stage. 

\sara{
In high-energy physics, for example, data from the Large Hadron Collider are often aggregated into discrete energy or mass bins and employed for searches for new physics beyond the Standard Model; of particular interest are thinly populated high-energy bins in which the new phenomena are often expected to appear \citep[e.g.,][]{belforte}. In astroparticle physics, the Fermi Large Area Telescope \citep{atwood} detects gamma rays from pulsars,  active galactic nuclei, and supernova remnants. Individual photon events are binned into three-dimensional cubes, corresponding to spatial and energy bins. In all these settings, the data consists of Poisson-distributed photon counts observed over many bins. 
}
Hence, the goal is to estimate and validate the hypothesized mean function of a Poisson process defined over a discretized space.

Extensions to the multinomial scheme or situations in which the cells are defined by categorical variables are also possible in the proposed setup.

\subsection{An analysis of data from the Chandra $X$-ray observatory} 
\label{sec:chandra_intro}
In this manuscript, we consider data from the Chandra $X$-ray observatory to demonstrate the practical implications of our findings in scientific research and illustrate their implementation.

The Chandra $X$-ray Observatory is an $X$-ray telescope orbiting the Earth since its launch in 1999. It captures and probes $X$-rays from astronomical sources over a wide wavelength range. The analysis of $X$-ray spectra obtained with Chandra provides information about the chemical composition, temperature, distance from the Earth, and other features of astronomical sources.    
Below, we consider an $X$-ray spectrum relevant to the study of RT Cru, a star belonging to a rare class of $X$-ray-emitting symbiotic systems. The origin of these systems remains poorly understood, making RT Cru a target of significant astronomical interest. 

While previous studies made important contributions in the study of RT Cru \citep[e.g.,][]{luna,vinay}, many questions have remained unaddressed due to substantial background contamination in the $15-20\r{A}$ (Angstroms) range. In astronomy, the term `background' refers to the combined signal of all astrophysical sources, which are not those we aim to study.
When analyzing data with a significant background component, accurate modeling of its distribution is crucial. For RT Cru, \citet[][]{zhang} conducted a comprehensive statistical analysis to validate and recalibrate proposed background models using `source-free spectra' -- that is, data collected in an off-source region of the detector. Our present analysis examines a subset of the data used in their study; it consists of a source-free spectrum in the $14.6-17.4\r{A}$ range. 
\sara{Given the high resolution of the detector, the original dataset is unbinned; however, since the study of high-resolution $X$-ray spectra is typically conducted by binning the data  \citep[e.g.,][]{luna,karovska2010,bonamente}, we choose to discretize the spectrum in 750 bins. As shown in Fig. \ref{fig:RTCru_spectrum}, the bins average approximately seven counts, with many containing fewer than five counts. Such a choice allows us to demonstrate how the methods proposed in this manuscript can be used to test models for both low-resolution and sparsely binned high-resolution spectra while warning practitioners about their limitations. }

For the data considered, \citet[][]{zhang} provided strong evidence that the hypothesized uniform background model should be rejected, and a linear correction for it was recommended. \sara{The present manuscript shows that many classical goodness-of-fit statistics for grouped data fail to detect departures from uniformity. It establishes the theoretical foundations underlying the general nature of the phenomenon. It also proposes new tools to overcome this deficiency. }

\section{Modeling framework}
\label{sec:modeling}
To formally define the quantities introduced in Section \ref{sec:introduction},  consider a Poisson process $\NT(A),A\subseteq \X, $ with mean 
\begin{equation}
\label{eqn:LambdaT}
E[\NT(A)]=T\Lambda_\beta (A)= T \int_A \lambda_\beta(x) dx,
\end{equation}
\sara{
where the parameter $T$  determines the size of the sample  --  that is, in our context, `large samples' means $T\rightarrow \infty$. The function $\lambda_\beta$  describes the spread of events over $\X$; it depends on $\beta\in \Real^{p-1}$ but not on $T$. We assume $\lambda_\beta$ is piecewise continuous,  bounded away from zero, and, without loss of generality, we can further assume it is a probability density function in $x$, for all $\beta\in \Real^{p-1}$.}

\sara{In practice,   $\beta$ is often unknown or unspecified. Moreover,  from a technical standpoint, it is convenient to treat  $T$ as an unknown parameter to be estimated, along with $\beta$. Since $T$ may grow without limit, we estimate it in relative terms by considering $c\sim \frac{T}{K}$.  Therefore,  $\theta=(c,\beta)^{\sf T}\in \Theta\subseteq \Real^{p}$ is the focus of estimation, and we assume $\lambda_\beta$ satisfies the usual local asymptotic normality (LAN) requirement  (see Definition 7.14 in \cite[][p.104]{vandervaart})}.

Let $\X$ be a finite union of  $d$-dimensional hypercubes. When the region of interest is not rectangular but merely a bounded set in $\Real^d$, we can cover it with a finite union of cubes, $\mathcal{R}$, and assume $\lambda_\beta(x) = 0$ for all $x \in \mathcal{R} \setminus \X$.
In other words, we can construct a grid on $\X$ composed of $K$ disjoint squares, or bins, denoted by $\{\Deltax\}_{k=1}^K$ and centered at points $x_k$. Assume all the bins have the same volume
\[\delta=\frac{|\X|}{K}.\] 

\sara{
The observed frequencies  correspond to the increments of the process $\NT$ over each bin, i.e.,
\[\nut=\NT(\Deltax),\quad \text{for $\klist$}, \]
where the dependence of $\nut$ on $T$  is dropped to simplify the notation.  The collection of the frequencies $\{\nu(x_k)\}_{k=1}^K$, which, from some point of view, one would
like to call a Poisson brush, is an unusual object. Indeed, $\nu(x_k)$ counts the number of jump
points of $N_T$ in the bin $\Delta[x_k]$; yet, as the bin shrinks, these points do not vanish because $T$ increases at the same time; but where do the limiting jump points live? The answer, we believe, is associated with the notion of fold-up derivatives of shrinking sets and the local Poisson processes (cf.
\cite{khm2007},  \cite{khm2008}). However, we will not digress in this direction here.}

\sara{
The expectation of each $\nut$ is:
\begin{equation}
    \label{eqn:asympt_eq}
    \mtheta=T\int_{\Deltax}\hspace{-0.25cm}\lambda_\beta(x) dx\sim T\delta \lambda_\beta(x_k)\sim c|\X| \lambda_\beta(x_k),
\end{equation}
in which the first \sara{asymptotic} equivalence follows from the Lebesgue differentiation theorem \citep{rudin87} and the second follows from  \eqref{eqn:A2}. Let $m_{_\theta}$ be the stepwise constant function defined as
 \[\mthetax=\mtheta\quad\text{for all $x\in\Deltax$},\]
converging to $c|\X|\lambda_\beta(x)$ by \eqref{eqn:asympt_eq}. Hereinafter, we denote with $\Etheta[\cdot]$  the expectation taken with respect to a Poisson random variable, $\nutx$, with mean $\mthetax$, and we let $P(\cdot|\mthetax)$ be its distribution.}

\section{Divisible statistics as linear functionals based on the same random measure} 
\label{sec:divisible}
\sara{
As described in Section \ref{sec:introduction}, classical goodness-of-fit statistics for testing  \eqref{eqn:test} fall under the class of divisible statistics defined as in \eqref{eqn:classical_ds}. This definition, however, does not include other important statistics encountered in the analysis of binned data, especially in the context of estimation.
We consider a more general class of divisible statistics with summands of the form
$\gtheta(x_k,\nut)=g(x_k,\nut,\mtheta)$ where each term may explicitly depend  on the variable $x_k$. Hence, $\gtheta(x,z)$ constitutes the `inhomogenous' version of $g(z,\mthetax)$.  For example, when different weights are assigned to different bins, we may choose
 \begin{equation}
 \label{eqn:omegag}
\gtheta(x,z)=\omega(x)g(z,\mthetax),
\end{equation}
leading back to the `homogeneous' divisible statistics  in \eqref{eqn:classical_ds}  when $\omega(x)=1$. 

Denote with $\mutheta$ the measure:
$$\mutheta(A,z) = \frac{1}{K}\sumk \mathds{1}_{\{x_k\in A\}}P(z| \mtheta),$$ 
and use as supply of functions $\gtheta$   the Hilbert space:
\begin{align*}\mathcal{L}_2 (\mutheta) &= \{ \gtheta:  \Etheta[ \gtheta(x,\nutx)]=0\text{ for all }x\in \mathcal{X}, \quad \| \gtheta\|<\infty\}\\
\text{with}\quad  \|\gtheta\|&=\sqrt{\langle \gtheta,\gtheta\rangle}\quad\text{and}\quad \langle \gtheta,\gtheta'\rangle = \int \hspace{-0.2cm} \int \hspace{-0.05cm} \gthetaxy \gtheta'(x,z){\mutheta}(dx,dz).
\end{align*} 
To bring unification to the theory, we propose a revision of the traditional definition of divisible statistics to include any statistic that can be expressed as a linear functional from the same random measure. 
\begin{definition}
\label{def:divisible}
Given a function $\gtheta\in \mathcal{L}_2 (\mutheta)$, a divisible statistic defined by $\gtheta$  is
\begin{equation}
\label{eqn:divisible_stat_general}
 \vtheta(\gtheta)=\int \hspace{-0.2cm} \int \hspace{-0.05cm} \gthetaxy \vtheta(dx, dz)=\frac{1}{\sqrt{K}}\sumk \gthetak
\end{equation}
where $\vtheta$ is the random measure
\begin{equation}\label{eqn:vktheta} 
\vtheta(A,z)=\frac{1}{\sqrt{K}}\sumk\mathds{1}_{\{x_k\in A\}}\hspace{-0.05cm}\Bigl[\mathds{1}_{\{\observed \le z\}}-P(z | \mtheta)\Bigl],\quad A\subseteq \X,z\in \NN.
\end{equation}
\end{definition}
The integral in \eqref{eqn:divisible_stat_general} is a random Stieltjes integral, and $\vtheta(\gtheta)$, as a function of $\gtheta$, is a direct analog of the function-parametric empirical process \citep[e.g.,][Ch 19]{vandervaart}.
For instance, estimating equations (to be introduced in Section \ref{sec:estimators}) are divisible statistics specifying as in \eqref{eqn:divisible_stat_general} equated to zero, and $\gtheta$ can be chosen as in \eqref{eqn:omegag}, with $\omega(x)$ depending on $\theta$. 
Partial sums can be constructed similarly by choosing $\omega(x)=\mathds{1}_{\{x\in A\}}$.  
For fixed $A$ and $z$,  $\vtheta(A,z)$ is itself a divisible statistic, the localised (and centered) cumulative spectral statistic.  It counts the number of bins in  $A$ with frequencies $\nut$ less than or equal to $z$ while centering the random terms,  $\mathds{1}_{\{\observed \le z\}}$,  by their expectation. 
}

\sara{
 To better understand the connection between the random measure $\vtheta$ and the  classical empirical process for i.i.d. observations, rewrite \eqref{eqn:vktheta}  as  
\begin{equation*}
\label{eqn:divisible_stat_xi2}
\vtheta(A,z) =\sqrt{K} \biggl[ \frac{1}{K} \sumk\mathds{1}_{\{x_k\in A\}}\mathds{1}_{\{\nut \le z\}}- \mutheta(A,z)\biggl]. 
\end{equation*}
The first average in $k$ serves a purpose similar to that of the empirical distribution function.
The expected value of such an empirical distribution, given the collection of $x_k$, is the measure $\mutheta(A,z)$.
As $K\rightarrow \infty$, the  marginal $\muk(A)=\muk(A,\infty)$ converges to  $\mu(A)=|A|/|\X|$, the normalized Lebesgue measure, or uniform distribution, on $\X$ and $\mutheta(A,z)$ converges weakly to   
\begin{equation}
\label{eqn:mutheta}
\mu_{_\theta}\hspace{-0.05cm}(A,z) =  \frac{1}{|\X|} \int_A P(z | \mthetax ) dx.
\end{equation}
For example,  let
\begin{equation}
\label{eqn:sigma2}
\sigma^2_{\gtheta}=\int \hspace{-0.1cm}\Etheta\bigl[\gtheta^2(x,\nutx)\bigl] \mu(dx);
\end{equation}
as $K\to\infty$, we have
\begin{align*}
||\gtheta||^2=\int \hspace{-0.2cm}\int \hspace{-0.1cm}\gtheta^2(x,z) \mutheta(dx,dz)=\hspace{-0.1cm}\int \hspace{-0.1cm}\Etheta\bigl[\gtheta^2(x,\nutx)\bigl] \muk(dx)\to \hspace{-0.1cm}\int \hspace{-0.2cm} \int \hspace{-0.05cm} \gtheta^2(x,z) \mu_{_\theta}\hspace{-0.05cm}(dx, dz)=\sigma^2_{\gtheta}.
\end{align*}}

\sara{
The stochastic representation in \eqref{eqn:divisible_stat_general} unfolds the intrinsic nature of all divisible statistics: no matter how non-linear they are in $\nu(x_k)$, they are all linear functionals of the random measure $\vtheta$. 
Seemingly different statistics, such as Pearson's $\chi^2$ statistic and the spectral statistic, can be expressed as the values of the same function-parametric process $\vtheta(\gtheta)$. They are, therefore, not so different and are subject to the same treatment. 

The representation in \eqref{eqn:divisible_stat_general} and the central limit theorem together provide the Gaussian limiting distribution of any divisible statistics. As a reference, we formalize such a statement in Proposition \ref{prop:gaussian}.
\begin{proposition}
\label{prop:gaussian}
If $\gtheta$ is Riemann integrable in $x$ and $m_{_\theta}$, and $\sigma^2_{\gtheta}<\infty$, 
then, under \eqref{eqn:A2} and if $H_0$ is true,  $$\vtheta(\gtheta)\xrightarrow[]{d} N(0,\sigma_{\gtheta}^2).$$
\end{proposition}
For the centered Pearson's statistic
\begin{equation}
\label{eqn:pearson}
\frac{1}{\sqrt{K}}\sumk \biggl[\frac{(\nut-\mthetax)^2}{\mthetax}-1\bigg]
\end{equation}
Proposition \ref{prop:gaussian} implies that, in the asymptotic regime described by \eqref{eqn:A2},  \sara{the usual asymptotic normality of the $\chi^2$ distribution with a large number of degrees of freedom} no longer applies: in the limit,  \eqref{eqn:pearson} is normally distributed with zero mean under $H_0$, and variance $\|\gtheta\|^2=2+\int \frac{1}{\mthetax}\mu(dx)$.}

\section{Divisible statistics with estimated parameters} 
\label{sec:estimators}
When testing the validity of a given model $m_{_\theta}$ without a prescribed value of $\theta$,  one would need to replace the latter with an estimator. Therefore, it is useful to study how parameter estimation affects the structure of divisible statistics.

Let $\btheta$ be a $p$-dimensional vector function  with linearly independent components in ${\mathcal L}_2(\mutheta)$ and consider the system of estimating equations
\begin{equation}
\label{eqn:general_est_eq}
\vthetahat(b_{_{\widehat{\theta}}}\hspace{-0.02cm}) = 0.
\end{equation}
For example, when using Maximum Likelihood Estimation (MLE), $\btheta$ is equal to the score function
\begin{equation*}
\label{eqn:score2}
\psitheta(x,z)=\frac{\partial}{\partial \theta}  \ln p(z | \mthetax)=
\frac{\mdotx}{\mthetax}(z - \mthetax),
\end{equation*}
with $p(z| t)$ denoting the Poisson probability mass function with rate $t$ and $\mdotx=\frac{\partial}{\partial \theta}\mthetax$. When estimating $\theta$ via least squares, the minimization
\[ \thetahat=\arg\min_{\theta}\sumk(\nut-\mtheta)^2,\]
leads to
$\bxy=-2\mdotx(z-\mthetax)$.

\sara{Hereafter, we assume that the following expansion is valid for  $\gtheta\in {\mathcal L}_2 (\mutheta)$:
\begin{align}
\label{eqn:linearize_g00}
\vthetahat(\gthetahat) = \vtheta(\gtheta) - \sqrt{K}(\wh\theta - \theta)^{\sf T}\langle \gtheta,\psitheta\rangle+ o_{_P}(1)
\end{align}
Sufficient conditions for the validity of \eqref{eqn:linearize_g00}, along with the proof, are provided in the Supplementary Material. Such conditions are direct analogues of the classical regularity conditions used in the asymptotic theory of statistical inference for parametric families \citep[][Ch.33]{cramer}, \citep[][ Ch.5]{vandervaart}.

Proposition \ref{prop:projection} below uncovers the existence of a projection operator behind most routine parameter estimation procedures. Such a projection simplifies our understanding of the asymptotic behaviour of $\vthetahat(\gthetahat)$ and opens up less familiar aspects of this behaviour. 
\begin{proposition}
\label{prop:projection}
Assume, in addition to the conditions of Proposition \ref{prop:gaussian},  that\eqref{eqn:linearize_g00} holds. Then
\begin{equation}
\begin{split}
\label{eqn:proj1}
\vthetahat(\gthetahat) = \vtheta\bigl(\Pi \gtheta\bigl)+ o_{_P}(1)\quad
\text{with}\quad \vtheta\bigl(\Pi \gtheta\bigl)=\vtheta(\gtheta) - \langle \gtheta,\psitheta^{\sf T}\rangle\langle \btheta, \psitheta^{\sf T}\rangle^{-1} \vtheta(\btheta)
\end{split}
\end{equation}
where the linear transformation
\begin{equation}
\label{eqn:projectionB}
\Pi \gthetaxy=\gthetaxy - \langle \gtheta,\psitheta^{\sf T}\rangle\langle \btheta, \psitheta^{\sf T} \rangle^{-1} \bxy
\end{equation}
is a projection of $\gtheta$ parallel to the function $\btheta$, which defines the estimating equations, and orthogonal to the score function $\psitheta$,   i.e.,
\begin{equation}
\label{eqn:projection1} 
\langle \Pi \gtheta,\psitheta\rangle= 0\quad \text{and} \quad \Pi \bxy = 0 .
\end{equation}
The projector $\Pi$ is orthogonal when $b_{_\theta}=\psi_{_\theta}$.  

Moreover, denote with $\sigma^2_{\Pi\gtheta}$ the limit of $\|\Pi\gtheta\|^2$. If $\int \hspace{-0.2cm}\int \btheta^2(x,z)\mu_{_\theta}(dx,dz)<\infty$, then $\sigma^2_{\Pi\gtheta}<\infty$  and,  under $H_0$, 
\begin{equation*}
\label{eqn:Gaussian}
\vthetahat(\gthetahat)\xrightarrow[]{d} N(0,\sigma^2_{\Pi\gtheta}).
\end{equation*}
\end{proposition}}

\sara{
\begin{proof}
The components of $\btheta$ belong to ${\mathcal L}_2(\mutheta)$; thus, the asymptotic representation in  \eqref{eqn:linearize_g00} can also be applied to the left-hand side of the estimating equations in \eqref{eqn:general_est_eq}. By equating the resulting expansion to zero, we obtain,
\begin{equation*}
\label{eqn:thetahat}
\begin{split}
\sqrt{K}(\wh\theta - \theta)= \langle \btheta, \psitheta^{\sf T}\rangle^{-1} \vtheta(\btheta) +o_{_P}\hspace{-0.05cm}(1)
\end{split}
\end{equation*}
and substituting the right-hand side of this equality into \eqref{eqn:linearize_g00} gives \eqref{eqn:proj1}. 

One can show that  $\Pi $ is a projection operator by verifying that 
$\Pi \Pi  \gtheta=\Pi \gtheta$; whereas, the conditions in \eqref{eqn:projection1} can be easily verified algebraically. 

Finally, since both $\gtheta$ and $\btheta$ are square integrable with respect to $\mu_{_\theta}$, the limits of $\langle \gtheta,\psitheta^{\sf T}\rangle$, $\langle \btheta, \psitheta^{\sf T}\rangle$, $\|\gtheta\|$, and $\|\btheta\|$ exist; hence,  $\sigma^2_{\Pi\gtheta}<\infty$ and the stated Gaussian distribution  follows from \eqref{eqn:proj1} and Proposition \ref{prop:gaussian}. \qed
\end{proof}

The projection in \eqref{eqn:proj1}, characterizing divisible statistics with estimated parameters, has implications of both practical and theoretical interest. 
For instance, it enables us to express  $\vthetahat(\gthetahat)$, with estimated parameter,  as $\vtheta$ evaluated at the function $\Pi \gtheta$ with fixed parameter. The latter belongs to the class of divisible statistics and is, therefore,  asymptotically Gaussian (see Section \ref{sec:chandraI} for an example). As formalized in Proposition \ref{prop:powerMLE} in Section \ref{sec:shift}, however, the random variable  $\vtheta(\Pi \gtheta)$ should be stochastically smaller than   $\vtheta(\gtheta)$, and thus, substituting $\widehat{\theta}$ in $\vtheta$ changes the entire class of divisible statistics, making them stochastically smaller. 
As described in Section \ref{sec:recovering_distr_free}, the existence of $\Pi$ is also the key fact that enables the construction of distribution-free tests, with known asymptotic distribution.

Let us now analyse the structure of $\Pi$.} The second equality in \eqref{eqn:projection1} tells us that $\btheta$ belongs to the kernel of the operator $\Pi$. When $\btheta\neq\psitheta$, however, the kernel of $\Pi$ does not coincide with the orthogonal complement of its image.  Hence, in general, the projection is not orthogonal, but, as noted in Proposition \ref{prop:projection}, it is when the parameters are estimated via MLE. For the latter, the leading term on the right-hand side of  \eqref{eqn:proj1} can be expressed as
\begin{equation*}
\label{eqn:projection_MLE}
\vtheta(\Pi\gtheta)=\vtheta(\gtheta) - \sum_{j=1}^p\langle \gtheta,s_j\rangle \vtheta(s_j),
\end{equation*}
in which $\psithetatildej$ denotes the  $j$-th coordinate of the orthonormalized score function
\begin{equation*}
\label{eqn:normalized_score}
s(x,z)=\langle \psitheta, \psitheta^{\sf T}\rangle^{-1/2}\psitheta(x,z)
\end{equation*}
and $\langle \psitheta, \psitheta^{\sf T}\rangle^{-1/2}$ is the inverse square root of the Fisher information matrix.

\sara{In real data analyses, statisticians would typically choose a function $g$ for goodness-of-fit testing, and, in some instances, they may also want to, or have to, use the same function to estimate the unknown parameters. In this case, we may choose the function $\btheta$ defining the estimating equations in \eqref{eqn:general_est_eq} equal to the right-hand side of \eqref{eqn:omegag} with the weighting function $\omega_{_\theta}$ depending on $\theta$.} That is, for a given choice of $g$, we consider 
the class of estimators defined by 
\begin{equation}
\label{eqn:est_eq_omega}
\vthetahat(\omega_{_{\thetahat}} g)=0.
\end{equation}
For instance, when $g$ gives the linear statistic, two members of such a class are the MLE and the least squares estimator with $\omega_{_\theta}(x)=\frac{\mdotx}{\mthetax}$ and $\omega_{_\theta}(x)=-2\mdotx$, respectively.   

The estimator with the minimum variance among those given by \eqref{eqn:est_eq_omega} can be constructed   by choosing $\omega_{_\theta}(x)$ equal to
\begin{equation}
\label{eqn:gamma_fun}
 \gamma_{_\theta}(x) = \frac{\Etheta \bigl[g(\nutx,\mthetax)\psitheta(x,\nutx)\bigl]}{\Etheta[g^2(\nutx,\mthetax)]}.
\end{equation}
This result is formalized in the next proposition.  

\begin{proposition}
\label{prop:optimality}
Let $\thetahat$ be a consistent root of the estimating equations defined in \eqref{eqn:est_eq_omega} 
and let $u$ be a vector in $\Real^p$. Under the conditions of Proposition \ref{prop:projection}, for a given choice of $g$, the asymptotic variance
of  $u^{\sf T} \sqrt{K}\bigl( \thetahat - \theta \bigr) $ is the smallest when $\omega_{_\theta}=\gamma_{_\theta}$.
\end{proposition}
\begin{proof}
By the definition of $\gamma_{_\theta}$ we have that
$\langle \omega_{_\theta} g,  \psitheta \hspace{-0.05cm}^{\sf T}\rangle= \langle \omega_{_\theta} g, g \gamma_{_\theta}\hspace{-0.05cm}^{\sf T}\rangle$. Therefore, depending on the weighting function used, the asymptotic representation of $\thetahat$ is given by either 
\begin{equation*}
\begin{split}
\label{eqn:asympt}
\sqrt{K}(\wh\theta - \theta)&\sim \langle \omega_{_\theta} g, g  \gamma_{_\theta}\hspace{-0.05cm}^{\sf T}\rangle^{-1} \vtheta(\omega_{_\theta} g)\\
\text{or}\quad&\sim \langle \gamma_{_\theta} g ,  g\gamma_{_\theta}\hspace{-0.05cm}^{\sf T}\rangle^{-1}  \vtheta(\gamma_{_\theta} g).
\end{split}
\end{equation*}
Let us now consider the random variables 
\begin{align*}
V=  u^{\sf T}  \langle \omega_{_\theta} g,g   \gamma_{_\theta}\hspace{-0.05cm}^{\sf T}\rangle^{-1}   \vtheta(\omega_{_\theta} g) \quad\text{and}\quad W=  u^{\sf T} \langle \gamma_{_\theta} g ,  g\gamma_{_\theta}\hspace{-0.05cm}^{\sf T}\rangle^{-1}  \vtheta(\gamma_{_\theta} g).
\end{align*}
After some algebra, we obtain that their covariance is
$\Etheta[V W]=\Etheta [W^2]$. 
Hence,
\[ 0\leq\Etheta\bigl[(V - W)^2\bigl] =\Etheta [V^2] - \Etheta [W^2];\]
thus, the variance of $W$ is always  smaller  or equal to that of $V$. \qed
\end{proof}

When $g$ defines the linear statistic, we have that  $\gamma_{_\theta}g=\psitheta$.  When $g$ is not linear, for each $x$ fixed,  $\gamma_{_\theta}g$ gives the best mean square approximation of $\psitheta$ given $g$. 
For instance, when $g$ defines  the centred Pearson's $\chi^2$ statistic, $\gamma_{_\theta}g$  specifies as
\[\frac{\mdotx}{2 \mthetax + 1}\biggl[\frac{(z-\mthetax)^2}{\mthetax}-1\biggl].\]
\sara{Another relevant example arises in the so-called problem of empty boxes,  in which the frequencies $\{\nu(x_k)\}_{k=1}^K$ are not available; the user only knows which bins are empty and which are not \citep{GneHanPit07,domanski}. Yet, inference on $m_{_\theta}$ is needed. The empty boxes statistic is given by the function
$$\mathds{1}_{\{z= 0\}}-p(0|\mthetax)$$
But what should be the form of the estimating equations for $\theta$?  Proposition \ref{prop:optimality} provides the answer. In particular, $\gamma_{_\theta}g$ is  
\[\frac{\mdotx}{1-p(0|\mthetax)}\bigl[\mathds{1}_{\{z= 0\}}-p(0|\mthetax)\bigl]\]
and coincides with the maximum-likelihood equation for Bernoulli trials.}


\subsection{Testing background uniformity in sparse $X$-ray spectra via Pearson's $\chi^2$ test}
\label{sec:chandraI}
In the analysis of source-free $X$-ray spectra, the mean function is often assumed to be constant \citep[e.g.,][]{luna,bonamente}, that is, $\mthetax=c$   for all $x\in \X$. In this setting, the score function is $\psitheta(x,z)=\frac{1}{c}(z-c)$ for all $x\in \X$ and $c$ can be estimated via MLE, i.e.,  $\widehat{c}=\frac{\sumk\nut}{K}$. 

\sara{
The divisible statistic $\vthetahat(\gthetahat)$ chosen to test the hypothesis of uniformity of the background is the centered Pearson's statistic in \eqref{eqn:pearson} with $\mthetax$ replaced by $\widehat{c}$. In this case,
$$\quad \| \gtheta\|^2=2+\frac{1}{c},\quad||\psitheta||^2=\frac{1}{c},\quad\text{so that}\quad s(x,z)=\frac{(z-c)}{\sqrt{c}},\quad\text{and}\quad \langle \gtheta,s\rangle=\frac{1}{\sqrt{c}}.$$
Therefore, when the number of bins considered is sufficiently large, by Proposition \ref{prop:projection},  the asymptotic null distribution of Pearson's statistic with $c$ estimated is Gaussian with zero mean and variance $\|\Pi \gtheta\|^2=\| \gtheta\|^2-\langle \gtheta, s\rangle^2=2.$}

Let us now test the hypothesized uniformity of the background density using the source-free spectrum from Chandra in Fig. \ref{fig:RTCru_spectrum}. For these data,  $\widehat{c}=6.947$, and the value of Pearson's $\chi^2$ statistic observed is $-0.852$.  
The p-value based on the Gaussian approximation is $0.547$.
Thus, the test \sara{accepts the constant background model rather obviously.
When analyzed by \citet{zhang} using smooth tests \citep[cf.][Ch.14]{lehmann1986testing}, the same data led to a clear rejection of uniformity (with a p-value $0.0004$).} This divergence in conclusions is a good illustration of the phenomenon we describe from a theoretical standpoint in Section \ref{sec:spacehomogeneous}: \sara{in the asymptotic regime considered in our study,} Pearson's $\chi^2$ and many other divisible statistics have no asymptotic power against weak departures from uniformity.

\section{Contiguous alternatives and adequacy for goodness-of-fit}
\label{sec:alternatives}
In the context of goodness-of-fit, the usual class of alternatives considered for such a study is that of \emph{contiguous alternatives}.  The notion of contiguity, introduced by \citet{lecam}, allows us to define sequences of alternatives that converge to the null hypothesis as the sample size increases, thereby describing local deviations from $m_{_\theta}$ due to fainter and fainter signals. Yet, these alternatives remain distinguishable from the null hypothesis.
\begin{figure}[!h]
    \centering
        \includegraphics[width=50mm]{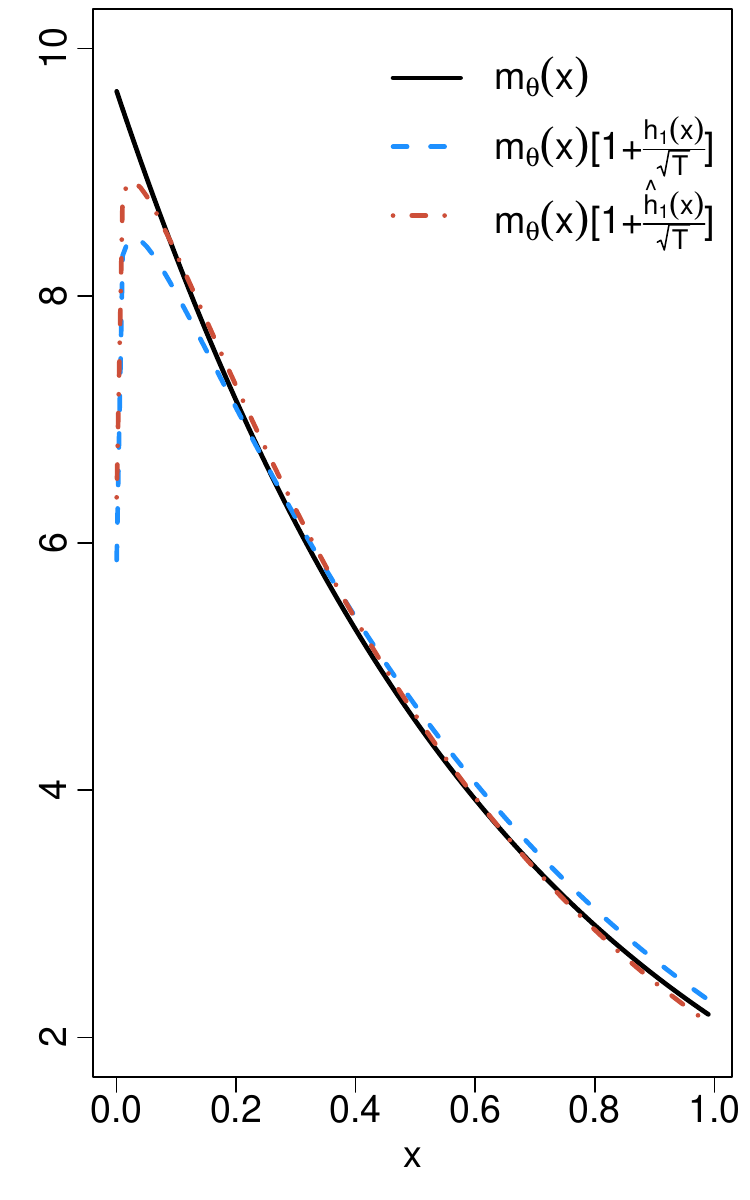}\hspace{-0.35cm}
        \includegraphics[width=50mm]{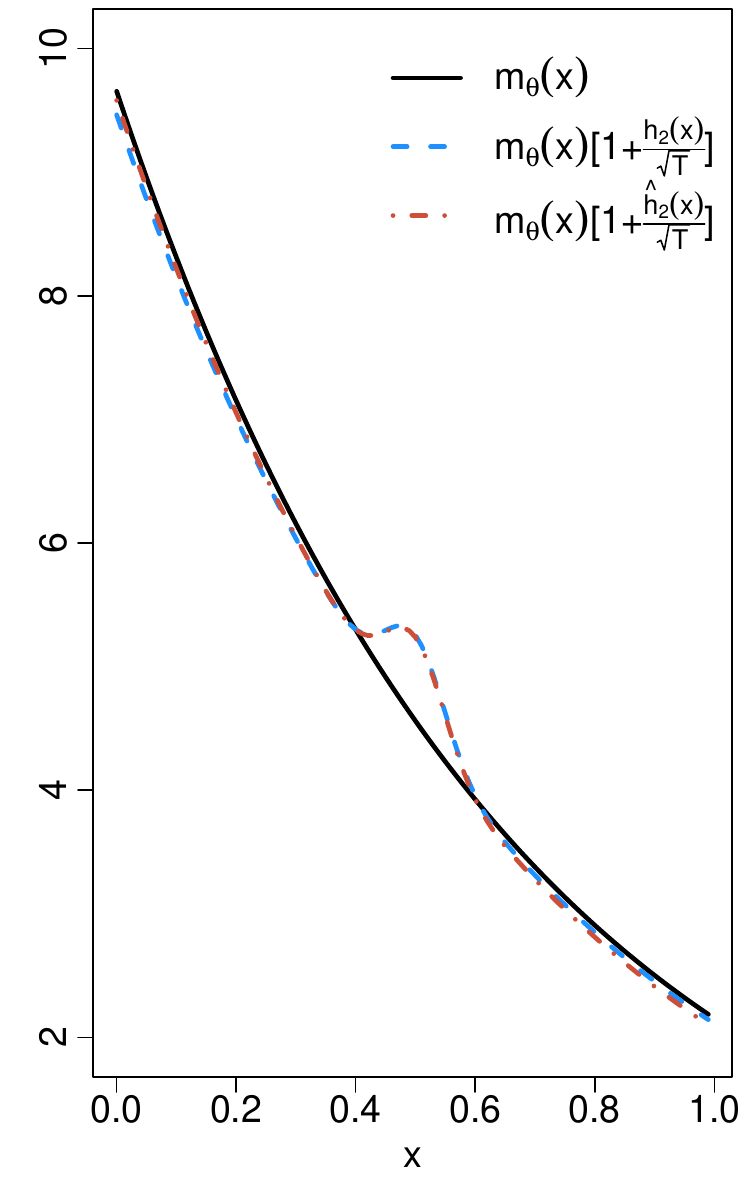}\hspace{-0.35cm}
               \includegraphics[width=50mm]{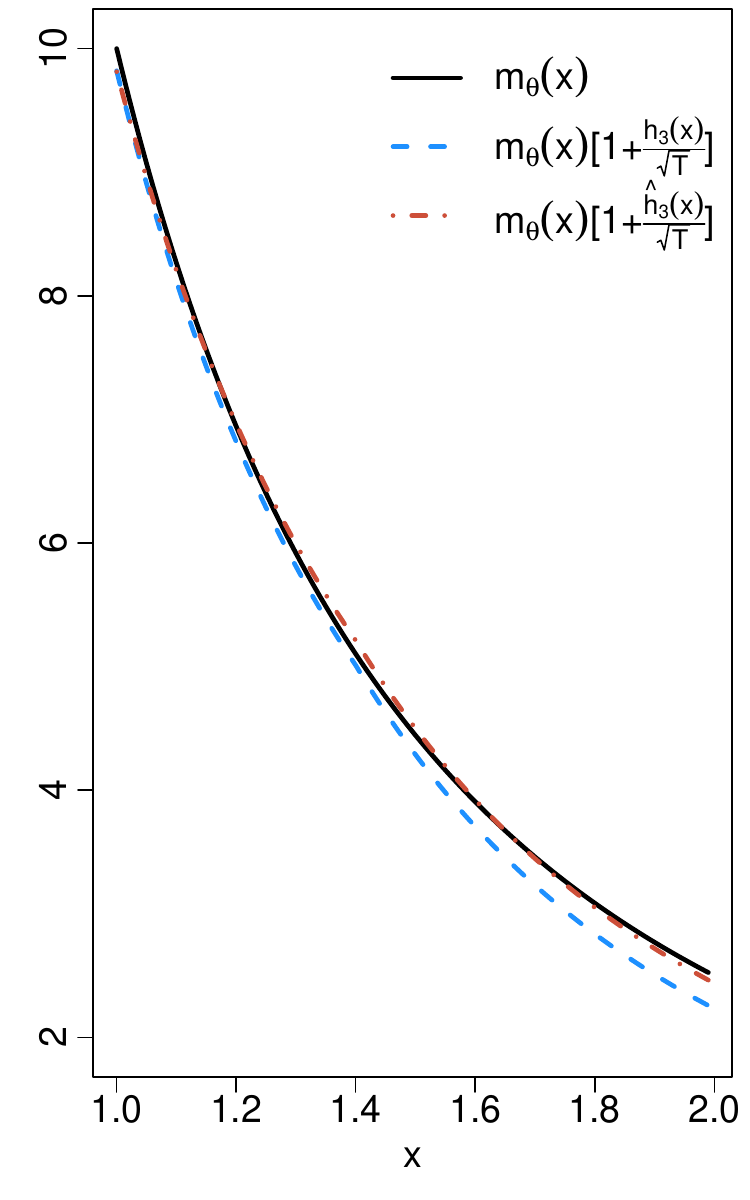}\\
       \vspace{-0.4cm}
    \caption{Mean functions under the null (black solid lines),  alternatives (blue dashed lines), and their component detectable by parametric tests obtained by replacing $h$ in \eqref{eqn:mtilde} with $\widehat{h}$ in \eqref{eqn:hhat} (red chained lines) for Example I (left panel), Example II (central panel), and Example III (right panel). }
    \label{fig1}
\end{figure} 
\sara{
\begin{definition}
\label{def:alternatives}
Consider the distribution function
$\Lambda_\beta(A)=\int_A \lambda_\beta(x)dx$.
When $\theta$ is known, we say that  
\begin{equation}
\label{eqn:mtilde}
\mtildex= \mthetax\biggl[1+\frac{ h_{_T}\hspace{-0.05cm}(x)}{\sqrt{T}}\biggl]
\end{equation}
defines  a sequence of \emph{contiguous alternatives} to  $m_{_\theta}$ if, 
for a given  sequence of functions, $h_{_T}\hspace{-0.05cm}\in L_2(\Lambda_\beta)$, we have
\begin{equation}
    \label{eqn:condition_contiguity}
    \limsup_{T\to\infty} \|h_{_T}\|^2_{\Lambda_\beta}=\limsup_{T\to\infty}\int h^2_{_T}(x)\Lambda_\beta(dx)<\infty.
\end{equation}
Moreover, if for some $h\in L_2({\Lambda_\beta})$,
\begin{equation}
    \label{eqn:condition_converging_cont}
    \|h_{_T}-h\|^2_{\Lambda_\beta}\rightarrow 0,\quad\text{as $T\rightarrow \infty$,}
\end{equation}
we say that $\mtildex$ defines a sequence of \emph{converging contiguous alternatives} to $m_{_\theta}$.
\end{definition}
}

The function $h$ defining a given sequence of alternatives can be interpreted as the functional direction from which $\widetilde{m}_{_{\theta,T}}$ approaches  $m_{_\theta}$ and $\|h\|^2_{\Lambda_\beta}$ can be taken equal to one without loss of essential content. Sequences of $\widetilde{m}_{_{\theta,T}}$ approaching $m_{_\theta}$ from the same direction can be considered equivalent and can be labeled by functions $h$ in the unit ball in $L_2(\Lambda_\beta)$.

When testing parametric hypotheses, the class of converging contiguous alternatives is much wider compared to the case of simple hypotheses. Such a class may seem to include all sequences of alternatives indexed by functions $h$ in the unit ball in any of the spaces $L_2(\Lambda_\beta)$, $\beta\in \Real^{p-1}$.  However, this is not quite true. Infinitesimal changes along the tangent spaces of the parametric family $\{\Lambda_\beta, \beta \in \Real^{p-1}\}$ are now not relevant. As the projection argument of Proposition \ref{prop:projection} and Proposition \ref{prop:powerMLE} in Section \ref{sec:shift} show, such changes are not detectable by statistics with estimated parameters.  \sara{ Hence, we extend our Definition \ref{def:alternatives} as follows:
\begin{definition}
\label{def:alternatives_par}
When testig parametric hypotheses, in addition to \eqref{eqn:mtilde}-\eqref{eqn:condition_converging_cont}, we further assume that
\begin{equation}
\label{eqn:orthogonal}
\begin{split}
\int h(x) \Lambda_\beta(dx) = 0\quad\text{and}
\quad\int h(x)\frac{\dot\lambda_\beta(x)}{\lambda_\beta(x)}  \Lambda_\beta(dx) = 0,
\end{split}
\end{equation}
with $\dot\lambda_\beta(x)=\frac{\partial}{\partial \beta}\lambda_\beta(x)$.
\end{definition}}
The first condition in \eqref{eqn:orthogonal} excludes sequences of alternatives that differ from  $\lambda_\beta$ only by a scaling factor and can be accounted for through the estimation of $c$. In other words, we consider alternatives that only affect \sara{the spread}  of the data over $\X$; as a result, the product of $\lambda_\beta$ and the term in the square brackets in \eqref{eqn:mtilde} integrates to one.
The second condition in \eqref{eqn:orthogonal} excludes alternatives indexed by functions $h\in L_2(\Lambda_\beta)$ which are linear combinations of the coordinates of $\frac{\dot{\lambda}_\beta(x)}{\lambda_\beta(x)}$. The latter describes \sara{small changes of $\lambda_\beta$ in $\beta$} \citep[cf.][]{amari,kass2011}.
Such deviations do not describe departures from the parametric family as such and are incorporated into the null through the estimation of $\beta$.

\sara{Alternatives for which both conditions in \eqref{eqn:orthogonal} hold are of the form
\begin{equation}
\label{eqn:hhat}
\widehat{h}(x)=h(x)-{\frac{\mdotx}{\mthetax}}^{\sf T}\Gamma^{-1}_{_\theta}\hspace{-0.05cm}\int{ \frac{\dot{m}_{_\theta}(y)}{m_{_\theta}(y)}} h(y) \Lambda_\beta(dy),
\end{equation}
with $\Gamma_{_\theta}\hspace{-0.05cm}=\bigintss \frac{\mdotx}{\mthetax}{\frac{\mdotx}{\mthetax} }^{\sf T}\Lambda_\beta(dx)=\frac{1}{c}\langle \psitheta,\psitheta^{\sf T}\rangle$.}
The three examples below illustrate how statistically interesting alternatives defined by functions $h$ may differ from alternatives defined by functions $\widehat{h}$ and from the null hypothesis. 

\textbf{\emph{Example I}}. Let $\Lambda_\beta$ be the distribution of an exponential random variable with rate $\beta=1.5$ truncated over the range $\X=[0,1]$. \sara{The alternative considered corresponds to the (truncated) gamma distribution with the same rate. Thus, the functional direction characterizing the alternative is the normalized score function for the shape parameter, i.e.,  
$$
h_1(x)=a\ln x+b;
$$
where $a=0.87$ and $b=1.21$ ensure that $h_1$ has a unit norm.  The function $h_1$ satisfies the first condition in \eqref{eqn:orthogonal}; hence, the alternative affects the spread of the data but has no impact on the expected sample size described by $c$. }

The left panel of Fig. \ref{fig1} shows the graphs of the models for the expected counts when $T=500$, $K=100$, and $c=5$. For this example, the alternative defined by $h_1$ (blue dashed line) and that given by $\widehat{h}_1$ (red chained line) are somewhat different from one another. In this case,  
\begin{equation}
\label{eqn:linear_part}
\frac{\dot{\lambda}_\beta(x)}{\lambda_\beta(x)}=\varphi(\beta)-x
\end{equation}
where $\varphi(\beta)$ is a constant depending on $\beta$. Thus, the component of $h_1$ collinear with \eqref{eqn:linear_part} describes deviations within the family of exponential distributions; whereas, the component $\widehat{h}_1$, orthogonal to \eqref{eqn:linear_part}, involves a logarithmic term and describes what is truly of interest: deviations within the family of gamma distributions with shape parameter different from one.

\textbf{\emph{Example II}}. In physics and astronomy, when a reliable description of the background is available, a common problem is the detection of a faint `bump-like' signal on top of it. Assume the background density, $\lambda_b$, to be the same truncated exponential null model described in Example I and let $\lambda_s(\cdot,x_0)$ be the density of the signal, here assumed to be Gaussian with mean $x_0=0.5$ and standard deviation $\sigma=0.05$. The density function for the alternative model, inclusive of both background and signal, is 
\begin{equation*}
    \label{eqn:mixtureA}
 (1-\eta_{_T}) \lambda_b(x)+\eta_{_T} \lambda_s(x,x_0),
\end{equation*}
where $\eta_{_T}\in \Real$ is the signal strength relative to $T$  --  that is, the expected number of signal events is $T\eta_{_T}$.   Assume $\eta_{_T}$ to be appropriately small, i.e.,  \sara{$\eta_{_T}\sim T^{-\frac{1}{2}}$}. 
Hence, the alternative (background+signal)  model defined by the above density can be rewritten as in  \eqref{eqn:mtilde}, with $h_{_T}(x)$ converging to
\begin{equation*}
\label{eqn:h_bump}
h_2(x)=a\bigl[e^{-\frac{1}{2}(\frac{x-x_0}{\sigma})^2+\beta x}-1\bigl]
\end{equation*}
in which the quantity in the square brackets corresponds to  $\frac{\lambda_s(x,x_0)}{\lambda_b(x)}-1$ and 
the constant $a=0.44$ ensures that the function $h_2$ has a unit norm. Such a function describes the changes in \sara{the spread}  due to the presence of the signal, whereas $\eta_{_T}$ quantifies the closeness of the null and the alternative models as a function of the sample size $T$, here set equal to 500. 

The central panel of Fig. \ref{fig1} shows that the alternatives defined by $h_2$ (blue dashed lines) and $\widehat{h}_2$ (red chained line) overlap for almost all $x\in[0,1]$. That is because, in this case, the part of $h_2$ collinear with \eqref{eqn:linear_part}, and responsible for deviations within the family of exponential distributions, is negligible.

\textbf{\emph{Example III}}. Luminosity functions can often be described using power laws and broken power laws. Discerning between the two is a common problem arising in $X$-ray astronomy \citep[e.g.,][]{broken_pl1,broken_pl2,broken_pl3}.   In statistical terms, a power law is equal to a Pareto type I distribution, while a broken power law can be constructed by introducing a cutpoint to induce a change in the slope. Let our null model correspond to a power law with slope $\beta=2$ and truncated over the interval $\X=[1,2]$.  The function
$$
h_3(x)=  a(\ln \xi-\ln x)\mathds{1}_{\{x\geq \xi\}}+b
$$
describes the direction from which a broken power law with a cutpoint at $\xi=1.4$ approaches the simple power law model. Set  $a=5.58$ and $b=-0.40$ to ensure that $h_3$ has norm one.  
The right panel of Fig. \ref{fig1} displays the graphs of the mean models under the null (black solid line), the alternatives defined by $h_3$ (blue dashed line), and $\widehat{h}_3$ (red chained line) when $T=500$ and $c=5$. In this example, 
$$
\frac{\dot{\lambda}_\beta(x)}{\lambda_\beta(x)}=\phi(\beta)-\ln x
$$
where  $\phi(\beta)$ is a constant depending on $\beta$. Comparing the above expression with that of $h_3$, it is apparent that most of the departure from the null hypothesis entailed by $h_3$ consists of deviations within the parametric family. As a result,  the statistically interesting part  $\widehat{h}_3$ that remains unaffected by the estimation of $\beta$ is rather weak and almost indistinguishable from the null.

In the literature, the use of contiguous alternatives typically requires the validity of  \eqref{eqn:condition_converging_cont}. This condition is sufficient for contiguity but not necessary. On the contrary, \eqref{eqn:condition_contiguity} is both sufficient and necessary  \citep[cf.][]{oosterhoff}. This implies that there exists a sub-class of alternatives that are contiguous even without \eqref{eqn:condition_converging_cont}. In this case, it is not possible to identify a direction from which $\widetilde{m}_{_{\theta,T}}$ approaches $m_{_\theta}$, and thus, the sequence of alternatives diverges or oscillates without limit around the null. These alternatives, interesting in themselves, cannot be detected via goodness-of-fit, neither in the continuous nor in the grouped data regime, but they remain discernible from the null through specifically adjusted empirical processes. For this reason, they have been called \emph{chimeric alternatives} \citep[cf.][]{khm98}.

Therefore, in the attempt to formally define the adequacy of a statistical test for goodness-of-fit, we rely on the class of converging contiguous alternatives. 
\begin{definition}{\rm (Adequacy for goodness-of-fit)}
\label{def:gof_propetry}
We say that a statistical test is adequate for goodness-of-fit if its power exceeds the significance level for all converging contiguous alternatives.
\end{definition}
While the expression `all converging contiguous alternatives' is formally correct, given how these alternatives are defined, it may seem an exaggeration. Indeed, all the alternatives considered here preserve the assumption that the observed process is Poisson, i.e., a point process with independent increments. Thus, the statements under the null and the alternative only apply to the intensity of this process. 

When testing simple hypotheses in a sparse regime,   tests based on a single divisible statistic, including Pearson's $\chi^2$, are known not to satisfy  Definition  \ref{def:gof_propetry} \citep[cf.][]{khm84}. As shown in the sections that follow, the problem persists when testing parametric hypotheses.

In the remainder of the manuscript, we will only consider alternatives within the class of converging contiguous alternatives; thus, we will simply refer to them as `alternatives'.

\section{On the behavior of divisible statistics under contiguous alternatives}
\label{sec:sec6}
\subsection{Distribution of divisible statistics under converging contiguous alternatives} 
\label{sec:shift}
The stochastic representation of divisible statistics in \eqref{eqn:divisible_stat_general} allows us to determine their power on the basis of the asymptotic behaviour of the shift they acquire under the alternatives. As noted in Section \ref{sec:divisible}, the statistics $\vtheta(\gtheta)$ are asymptotically Gaussian with zero mean. The limits of their first two moments under the converging contiguous alternatives are given in the following proposition. 
\begin{proposition}
\label{prop:magic_form}
Denote with $\Etilde[\cdot]$ and $\Vtilde[\cdot]$, respectively, the expectation and the variance taken under the assumption that each $\nut$ has Poisson distribution with expected value $\mtildexk$. Let $h$ be the functional direction defining the alternatives as in \eqref{eqn:mtilde}-\eqref{eqn:condition_converging_cont}. Under the conditions of Proposition \ref{prop:gaussian},
\begin{equation}\label{eqn:shift_g}
\Etilde[\vtheta(\gtheta)]\sim \frac{1}{\sqrt{c}} \int C(x;\gtheta)  h(x) \mu(dx),
\end{equation}
where
\begin{equation}
\label{eqn:Cg} 
C(x; \gtheta) = \Etheta [\gtheta(x, \nutx)(\nutx - \mthetax)].
\end{equation}
Moreover, $\Vtilde[\vtheta(\gtheta)]\sim \sigma^2_{\gtheta}$.
\end{proposition}
\begin{proof}
Recall that $\gtheta$ is assumed to be centered under the null. The expected value under the alternative is
$$
\Etilde[\vtheta(\gtheta)]= \sqrt{K} \int \sum_{z=0}^\infty \gtheta (x,z) \left [p(z|\mtildex) - p(z|\mthetax)\right]\muk(dx)
$$
where the difference in the square bracket is asymptotically small because $\mtildex - \mthetax$  is asymptotically small, while the Poisson probabilities $p(z|t)$ are regular in $t$. Therefore, one can approximate the difference in the square brackets using
\begin{align} 
\sqrt{K} \frac{\partial p(z|\mtheta)}{\partial \mtheta} (\mtildexk &- \mtheta) = (z-\mtheta)p(z|\mtheta)
\sqrt{K} \frac{\mtildexk - \mtheta}{\mtheta} \notag \\
&\sara{\sim} (z-\mtheta)p(z|\mtheta) \frac{1}{\sqrt{c}}h_{_T}(x_k).
\end{align}
The substitution of this approximation in the expression for $\Etilde[\vtheta(\gtheta)]$ leads to the result. (A formal justification of the approximation is given in the Supplementary Material.) As to the variance, the convergence $p(z|\mtildex) \to p(z|\mthetax)$ implies
$$ \Vtilde[\vtheta(\gtheta)]=\int \Etilde\bigl[\gtheta(x,\nu(x)) - \Etilde[\gtheta(x,\nu(x))]\bigl]^2 \muk(dx) \to \sigma^2_{\gtheta}.$$ 
\qed
\end{proof}
Since the asymptotic distribution of $\vtheta(\gtheta)$ under the alternatives differs from the null only in the mean, to assess the adequacy of divisible statistics for goodness-of-fit, we focus on their shift.

Assume that the representation in \eqref{eqn:proj1} is valid. According to Proposition \ref{prop:magic_form}, when $\thetahat$ solves the estimating equations in \eqref{eqn:general_est_eq}, the limiting shift of $\vtheta(\Pi\gtheta)$ becomes
\begin{equation}
\label{eqn:shift_projection}
\tilde \Etheta[\vtheta(\Pi \gtheta)] \sim \frac{1}{\sqrt{c}}\int C(x;\Pi \gtheta) h(x) \mu(dx).
\end{equation}
The analysis of when the right-hand side can be zero -- and, therefore, of when tests based on $\vthetahat( \gthetahat)$ have no \sara{asymptotic} power -- leads to four main findings. The first two are formalized in Proposition \ref{prop:powerMLE}, the others will be discussed in Sections \ref{sec:spacehomogeneous}-\ref{sec:linear}.
\begin{proposition}
\label{prop:powerMLE}
Under the conditions of Proposition \ref{prop:projection}:
\begin{itemize}
\item[(i)] The limiting shift of $\vtheta(\Pi\gtheta)$ under sequences of alternatives defined as in \eqref{eqn:mtilde}-\eqref{eqn:condition_converging_cont} is  
\begin{equation}
\label{eqn:shift_hhat}
\Etilde [\vtheta(\Pi \gtheta)] \sim\frac{1}{\sqrt{c}}\int C(x;\Pi \gtheta)\widehat{h}(x) \mu(dx),\end{equation}
with $\widehat{h}$  satisfying the conditions in \eqref{eqn:orthogonal}. Thus, tangential deviations defined by the vector function $\frac{\dot{m}_{_\theta}}{m_{_\theta}}$ are, asymptotically, not detectable.

\item[(ii)] If $\theta$ is estimated via MLE, \sara{then,  for any sequence of converging contiguous alternatives,} the right-hand side of \eqref{eqn:shift_hhat} reduces to
\begin{equation}
\label{eqn:shift_MLE}
 \frac{1}{\sqrt{c}}\int C(x; \gtheta) \widehat{h}(x) \mu(dx).
\end{equation}
Therefore, for sequences of alternatives defined by a function $\widehat{h}$, the limiting power of $\vthetahat(\gthetahat)$ is higher than or equal to that of $\vtheta(\gtheta)$.
\end{itemize}
\end{proposition}
\begin{proof}
From the definition of $C(\cdot; \gtheta)$ in \eqref{eqn:Cg}, it follows that
$$C(x;\Pi \gtheta)=\Etheta \bigl[\Pi\gtheta(x, \nutx)(\nutx - \mthetax)\bigl]$$
while from the definition of the projection in \eqref{eqn:projectionB} we know that $\Pi\gtheta$ is orthogonal to $\psitheta$, see \eqref{eqn:projection1}. If we now multiply both sides of the above equality by $\frac{\dot{m}_\theta}{m_\theta}$ and we integrate with respect to $\muk$, we obtain:
$$\int C(x;\Pi \gtheta) \frac{\mdotx}{\mthetax}\muk(dx)=\int \Etheta \bigl[\Pi\gtheta(x, \nutx)(\nutx - \mthetax)\bigl]\frac{\mdotx}{\mthetax}\muk(dx)=\langle \Pi\gtheta , \psitheta\rangle=0.$$
However, the functions $h$ and $\widehat h$ differ exactly by a linear combination of the coordinates of $\frac{\dot{m}_{_\theta}}{m_{_\theta}}$, see \eqref{eqn:hhat}. Therefore,
$$\int C(x;\Pi \gtheta) h(x)\muk(dx)=\int C(x;\Pi \gtheta) \widehat h(x)\muk(dx) $$
converging to the integral in \eqref{eqn:shift_MLE} as $K\rightarrow \infty$, hence (i) follows.


When $\theta$ is estimated via MLE, according to Proposition \ref{prop:projection},  $\Pi$ is an orthogonal projector parallel to the score function $\psitheta$. Moreover,
$$\int C(x;\psitheta) \widehat{h}(x) \mu(dx) \propto \int   \frac{\dot{m}_{_\theta}(x)}{m_{_\theta}(x)}  \widehat{h}(x) \Lambda_{\beta}(dx)=0. $$
Hence, for sequences of alternatives defined by a function $\widehat{h}$, the limit of both $\Etilde [\vtheta(\gtheta)]$ and $\Etilde [\vtheta(\Pi \gtheta)]$ is equal to \eqref{eqn:shift_MLE}.
Under the null hypothesis, the  variance of $\vtheta(\Pi\gtheta)$ is 
$$\|\Pi \gtheta\|^2  \le \|g\|^2 $$
and, from Proposition \ref{prop:magic_form}, it is asymptotically equal to $\Vtilde[\vtheta(\Pi\gtheta)]$. 
Therefore, in the limit, the signal-to-noise ratio for $\vthetahat(\gthetahat)$ is greater than or equal to that of $\vtheta(\gtheta)$ and, since both statistics are asymptotically Gaussian,  statement (ii) follows.  \qed
\end{proof}
As noted in Section \ref{sec:alternatives}, the first statement of Proposition \ref{prop:powerMLE} confirms that,  when the parameters are estimated, only alternatives of the form in \eqref{eqn:hhat} can be detected by divisible statistics \sara{-- that is, we have no power in detecting deviation within the parametric family since those are already captured by parameter estimation.}  The second statement suggests that, regardless of the choice of $\gtheta$, even when $\theta$ is known, estimating it via MLE can lead to higher power against alternatives for which $h=\widehat{h}$.   One could easily quantify the power gain induced by MLE by using the Gaussian approximation of divisible statistics. The validity of Proposition \ref{prop:powerMLE} (ii) in finite samples will be illustrated through a numerical example in Section \ref{sec:linear}.


\subsection{No single divisible statistic is adequate for goodness-of-fit}
\label{sec:loss_power}
The limiting shifts in \eqref{eqn:shift_g} and \eqref{eqn:shift_projection} are linear functionals in $L_2(\mu)$; 
hence, for any $h$ lying in the subspaces of $L_2(\mu)$ annihilated by these functionals, the limiting power will be equal to the significance level. 

For instance, for the (centered) Pearson's $\chi^2$ statistic in \eqref{eqn:pearson} we have that $C(x;\gtheta)=1$ for all $x\in \X$. Hence, when $\theta$ is estimated via MLE,  the shift in \eqref{eqn:shift_MLE} reduces to
\begin{equation}
\label{eqn:int_hdx}
 \frac{1}{\sqrt{c}}\int \widehat{h}(x)\mu(dx).
\end{equation}
Therefore,  Pearson's $\chi^2$ has no \sara{asymptotic} power against alternatives approaching the null from a direction $\widehat{h}$ orthogonal, with respect to the Lebesgue measure, to constant functions. 
Equation \eqref{eqn:int_hdx} also implies that no departures from uniformity can be detected by Pearson's statistic, a problem originally pointed out by \citet{chibisov65}. Unfortunately, this fallacy is not limited to the case of Pearson's $\chi^2$ but, as demonstrated in Section \ref{sec:spacehomogeneous},  it affects an entire subclass of divisible statistics.



\textbf{\emph{Examples I-III (continued).}} When $c$ and $\beta$ are estimated via MLE,  for all three examples, the power of Pearson's $\chi^2$, obtained by means of $100,000$ Monte Carlo replicates, is approximately equal to the significance level, here chosen to be $5\%$. 
 As noted above, when the parameters are estimated, only the alternatives defined by the functions $\widehat{h}_j$, $j=1,2,3$ (red chained lines in Fig. \ref{fig1}) are detectable by divisible statistics. In Example III, the effect of $\widehat{h}_3$ is rather faint. Thus, regardless of the choice of $\gtheta$, no divisible statistic is likely to detect such a deviation. In Examples I-II, both alternatives defined by $\widehat{h}_1$ and $\widehat{h}_2$  are distinguishable from the null. When using Pearson's $\chi^2$,  however, the limiting shift is  $-0.014$ in Example I and $-0.019$ in Example II, the standard deviations are, respectively, 1.406 and 1.416. This tells that,  despite $\widehat{h}_j $, $j=1,2$  being not so small, their integrals are, thereby leading to a loss of power.  

Although divisible statistics do not satisfy Definition \ref{def:gof_propetry}, in some situations scientists may be willing to sacrifice adequacy for goodness-of-fit to ensure high power against a few alternatives relevant to a given application. In the construction of such tests, the structure of the shift in \eqref{eqn:shift_projection} can be used to verify that, for a given choice of $\gtheta$, power is preserved in the directions defining such alternatives.

\subsection{On the impossibility of testing uniformity via \sara{$C$-homogeneous} divisible statistics}
\label{sec:spacehomogeneous}
In the previous sections, we saw some implications of the fact that the shift of $\vthetahat(\gthetahat)$  is a linear functional from $\widehat{h}$. In particular, from Proposition \ref{prop:powerMLE} (i), we have that this linear functional is not defined by the function $\gtheta$ but by  $C(\cdot;\Pi \gtheta)$. This leads to a phenomenon that does not have an analogy in the classical theory of empirical processes. What can now occur is the following: a non-zero function $\gtheta$ produces a non-zero divisible statistic, but if $C(x;\Pi\gtheta)=0$  for all $x\in\X$, then the shift of this statistic is zero for all alternatives.

One situation in which this phenomenon arises is when   $\vtheta(\gtheta)$  belongs to the class of  \emph{\sara{$C$-homogeneous}} divisible statistics defined as follows. 

\begin{definition} We say that $\vtheta(\gtheta)$ is a {\it \sara{$C$-homogeneous}} divisible statistic, if the function $C(\cdot;\gtheta)$ in \eqref{eqn:Cg} is constant.
\end{definition}
For example, regardless of the model being tested, both the (centered) Pearson's statistic  
and the statistic  
\sara{\begin{equation}
\label{eqn:linear_stat}
\frac{1}{\sqrt{K}}\sum_{k=1}^K\frac{\nu(x_k)-\mtheta}{\mtheta}
\end{equation}}
are \sara{$C$-homogeneous}. 

When $\lambda_\beta$ is uniform, the class of \sara{$C$-homogeneous} statistics is rather broad. In this case, the frequencies $\{\nu(x_k)\}_{k=1}^K$ are all identically distributed Poisson random variables with mean $\mtheta=c$ for all $k=1,\dots,K$. Therefore, if $\gtheta(\nu(x_k), \mtheta)= g(\nu(x_k), c)$ does not depend on $x_k$ explicitly, then these random variables are also identically distributed for all $k=1,\dots,K$ and thus  $C(\cdot;\gtheta)$ is constant.
This brings us to the following statement:
\begin{proposition}
\label{prop:unif}
Suppose \eqref{eqn:shift_g} is valid, $\lambda_\beta$ is the uniform density on $\X$, and $c$ is estimated with MLE. Then, for any \sara{$C$-homogeneous} divisible statistics,
$$\Etilde  [\vtheta(\Pi \gtheta)] \sim 0,$$
that is, no such statistic has any asymptotic power for any $h$.    
\end{proposition}
\begin{proof}
Since $C(\cdot;\gtheta)$ is constant and $\lambda_\beta(x)=|\X|^{-1}$ for all $x\in\X$, we have
$$ C(x;\Pi \gtheta) = C(x;\gtheta)-\int C(x;\gtheta) \muk(dy)  = 0.$$
It follows that, for any $h$, the limiting shift in \eqref{eqn:shift_projection} is zero.
  \qed
\end{proof}
\sara{Proposition \ref{prop:unif} tells us that, although uniformity is a relevant null hypothesis in many practical settings, under \eqref{eqn:A2}, $C$-homogeneous statistics are unable to detect small departures from it. }


\subsubsection{Unsuitability of \sara{$C$-homogeneous} statistics in testing constant background in $X$-ray spectra}
\label{sec:chandraIb}
When analyzing the source-free $X$-ray spectrum from Chandra introduced in Section \ref{sec:chandraI}, Pearson's statistic failed to reject the hypothesis of a constant background model. As a follow-up study, consider the weighted linear statistic
in \eqref{eqn:linear_stat} 
and the (centered) likelihood ratio -- known in astronomy as `Cash statistic' \citep{cash} -- defined by the function 
$$\gtheta(x,z)=z\ln z-\Etheta[\nutx \ln \nutx]-(z-\mthetax)(1+\ln\mthetax).$$
The importance of \eqref{eqn:linear_stat} will be discussed in the next section.
When testing the hypothesis $H_0:\mthetax=c$ for all $x\in\X$, for both statistics, $C(\cdot;\gtheta)$ is constant and $C(\cdot;\Pi\gtheta)$ is identically equal to zero. The p-values obtained by means of a parametric bootstrap simulation involving 100,000 replicates are $0.145$ for the weighed linear statistic and $0.987$ for the likelihood ratio. As expected,  neither statistic rejected the constant background model. 
\sara{
Indeed, the deviations of the true background model from uniformity could not be detected by any \sara{$C$-homogeneous} divisible statistics.} Nevertheless, as shown in Section \ref{sec:recovering_GOF}, they can be detected when using omnibus functionals of the process involving their partial sums.

\subsection{Every divisible statistic is dominated by a weighted linear statistic. }
\label{sec:linear}
The situation in which $C(\cdot;\Pi \gtheta)$ is identically equal to zero for non-zero choices of $\gtheta$ can be retrieved in most settings, not only in the case of \sara{$C$-homogeneous} statistics. 

To see how this can happen, consider the structure of the projection operator in \eqref{eqn:projectionB}. At first glance, choosing $\gtheta$ orthogonal to the score function -- i.e., such that $\langle \gtheta, \psitheta\rangle = 0$ -- may seem advantageous: for all such $\gtheta$, $\Pi \gtheta = \gtheta$, thus, statistics $\vthetahat(\gthetahat)$ will be asymptotically equivalent to $\vtheta(\gtheta)$ and, therefore, their asymptotic behaviour will not depend on whether we estimate $\theta$ or not.
However, the condition $\langle \gtheta, \psitheta\rangle = 0$ is achieved as soon as $\gtheta (x_k, \nut)$ and $\nut$ are uncorrelated for every given $x_k$, i.e.,
$$C(x_k, \gtheta) = \Etheta\bigl[\gtheta (x_k,\nut) (\nut - \mtheta)\bigl] = 0,$$
which then leads to a zero limiting shift of the statistic $\vthetahat(\gthetahat)$.
Let us understand the situation better. 
In particular, let us decompose a given $\gtheta$ into the parts orthogonal to and collinear with $\nutx -  \mthetax$:
\begin{equation}
\label{eqn:gparallel}
\gtheta^{\perp}(x,z) = \gtheta(x,z) - \frac{C(x, \gtheta)}{\mthetax}(\nutx -  \mthetax), \quad \gtheta^{\parallela}(x,z) = \frac{C(x, \gtheta)}{\mthetax}(\nutx -  \mthetax).     
\end{equation}
Since $C(x, \gtheta^\perp)=0$ for all $x\in \X$, unbinding $\vtheta(\gtheta)$ from the extra randomness due to   $\vtheta(\gtheta^\perp)$ allows us to preserve the asymptotic shift while simultaneously reducing the variance. As formalized in Proposition \ref{prop:linear}, the same holds true when $\theta$ is estimated via MLE. 
\sara{
\begin{proposition}
\label{prop:linear}
Assume that the asymptotic representation in \eqref{eqn:proj1} is valid and $\theta$ is estimated via MLE. Then, for any divisible statistic defined by a function $\gtheta$ we have that:
\begin{itemize}
\item[(i)] If $\gtheta$ is uncorrelated with $(z-\mthetax)$, i.e., $\gtheta\equiv \gtheta^{\perp}$, then,  for any sequence of converging contiguous alternatives,  $\vthetahat(\gthetahat)$ has no limiting power.
\item[(ii)] If $\gtheta\not\equiv \gtheta^{\perp}$, then,  for any sequence of converging contiguous alternatives, the limiting power of $\vthetahat(\gthetahat^{\parallela})$ is greater or equal than that of $\vthetahat(\gthetahat)$ with equality holding when $\gtheta^{\perp}$ is identically equal to zero.
\end{itemize}
\end{proposition}
}
\begin{proof}
 Recall that $\vthetahat(\gthetahat)$ is asymptotically equivalent to $\vtheta(\Pi \gtheta)$ and consider the expansion
$$\vtheta(\Pi \gtheta) = \vtheta(\Pi g^\perp_{_\theta}) + \vtheta(\Pi \gtheta^{\parallela}).$$
Now notice that, since $\langle g^\perp_{_\theta},\psitheta^{\sf T}\rangle =0$,
\begin{equation}\label{covar_c}
\begin{split}
C(x;\Pi g^\perp_{_\theta}) = C(x;g^\perp_{_\theta}) - \langle g^\perp_{_\theta},\psitheta^{\sf T}\rangle \langle \btheta, \psitheta^{\sf T}\rangle^{-1} C(x; \btheta)=0.
\end{split}\end{equation} 

This fact and  Proposition \ref{prop:powerMLE} imply that $\Etilde \bigl[\vtheta(\Pi g^\perp_{_\theta})\bigl]  \sim 0$ for all alternatives. Therefore, \sara{$\vtheta(\Pi g^\perp_{_\theta})$ has no limiting power, which proves (i).} At the same time, $\Etilde \bigl[\vtheta(\Pi \gtheta) \bigl]$ and $ \Etilde\bigl[\vtheta(\Pi \gtheta^{\parallela})\bigl]$ have the same limit, i.e.,
$$ \frac{1}{\sqrt{c}}\int C(x;\Pi \gtheta^{\parallela})\widehat{h}(x) \mu(dx) = \frac{1}{\sqrt{c}}\int C(x;g^{\parallela}_{\theta})\widehat{h}(x) \mu(dx) ,$$
where the equality follows from Proposition \ref{prop:powerMLE}(ii). 
At the same time, from the equality in \eqref{eqn:projectionB} we have $\Pi\gtheta^\perp=\gtheta^\perp$ and if $\theta$ is estimated via MLE, $\btheta=\psitheta$ and $\Pi$ is orthogonal, i.e., it is a self-adjoint projector. Therefore,
$$\langle\Pi \gtheta^\perp,\Pi\gtheta^{\parallela}\rangle=\langle \Pi^2\gtheta^{\perp},\gtheta^{\parallela}\rangle=\langle\gtheta^{\perp},\gtheta^{\parallela}\rangle=0.$$
Since $\Pi\gtheta^{\perp}$ and $\Pi\gtheta^{\parallela}$ are orthogonal, the limiting variance of $\vtheta(\Pi \gtheta^{\parallela})$ is 
$$ \| \Pi \gtheta^{\parallela}\|^2=\| \Pi \gtheta\|^2-\| \Pi \gtheta^{\perp}\|^2\leq \|\Pi \gtheta\|^2.$$
Hence, the asymptotic power of $\vthetahat(\gthetahat^{\parallela})$ always exceeds or is equal to that of $\vthetahat(\gthetahat)$, \sara{with  equality holding when $\gtheta\equiv\gtheta^\perp$ or $\gtheta\equiv\gtheta^{\parallela}$, which proves (ii).} \qed
\end{proof}
Let us now investigate with an example if, as stated in Proposition \ref{prop:linear}, weighted linear statistics are indeed better and if, as proved in Proposition \ref{prop:powerMLE}, estimation via MLE yields higher power.\\

\sara{
\textbf{\emph{Example IV.}} 
Let the hypothesized $\lambda_\beta$ be a normal density with mean $\beta$, and variance $\sigma^2$, truncated over the interval $\X=[0,a]$. 
Choose $a=1$, albeit the specific numerical value plays little role. Also, in the numerical calculations, we assume $\beta=0.5$ and $\sigma^2=0.04$.
To test our hypothesis, begin with Pearson's $\chi^2$ in \eqref{eqn:pearson}. Since for this statistic, $C(x_k;\gtheta)=1$ for all $x_k$, the `better' weighted linear statistic, $\vtheta(\gparallel)$ in \eqref{eqn:gparallel} specifies as in \eqref{eqn:linear_stat}.

To evaluate the power, we consider the most straightforward choice of the alternative: the same truncated normal density but with variance somewhat different from $\sigma^2$. It can be shown that the functional direction corresponding to alternatives with variance approaching $\sigma^2$ is given by
\label{eqn:h_normal}
\begin{equation*}
h_4(x)=\frac{b}{2\sigma^2}\biggl[\frac{(x-\beta)^2}{\sigma^2}-\int_0^a\frac{(t-\beta)^2}{\sigma^2}\Lambda_\beta(dt)\biggl]
\end{equation*}
for some constant $b$, here chosen equal to $1/\sqrt{2}$.
The score function of the density $\lambda_\beta$ is
$$\frac{\dot\lambda_\beta(x)}{\lambda_\beta(x)} = \frac{x-\beta}{\sigma^2}-\int_0^a\frac{t-\beta}{\sigma^2}\Lambda_\beta(dt)$$
and  one can see that $h_4$ satisfies both orthogonality conditions in \eqref{eqn:orthogonal} for $\beta=0.5$ -- in alignment with the normal model where the derivative of the log-likelihood with respect to the mean and the variance are orthogonal. 

The power of the two statistics considered is compared using a numerical simulation involving 100,000 Monte Carlo replicates with $K=100, c=5$,  and choosing the significance level to be $5\%$. When $\theta$ is known, the power of the Pearson's $\chi^2$ is approximately $7.16\%$; whereas, for $\vtheta(\gparallel)$, the power increases to $10.53\%$. As expected, if we free ourselves from the extra randomness due to $\vtheta(\gtheta^\perp)$, the power increases.
When $\theta$ is estimated via MLE, the power of the $\chi^2$ statistics is $7.38\%$ and that of $\vthetahat(\gparallelhat)$ increases to $14.09\%$ -- that is, compared to $\vtheta(\gtheta)$,  the combined effect of the use of $\gparallel$ and the MLE doubles the power. 

What was observed above in the case of Pearson's statistic is not an isolated instance. Consider, for example, the spectral statistic
\begin{equation}
\label{eqn:spectral}
\frac{1}{\sqrt{K}}\sumk [\mathds{1}_{\{\nut\leq q\}}-P(q|\mtheta)],\quad\text{with $q \in   \Real$,}
\end{equation}
for which $C(x;\gtheta)=-\mthetax p(q-1|\mthetax)$; thus, $\vthetahat(\gtheta^{\parallela})$ specifies as
\begin{equation}
\label{eqn:vgparallel_spectrum}
\frac{1}{\sqrt{K}}\sumk p(q-1|\mtheta)(\nut-\mtheta).
\end{equation}
Choose  $q=1$, $K=100$, $c=5$, and let the significance level be $5\%$. When $\theta$ is known, the simulated power of \eqref{eqn:spectral} obtained using $100,000$ Monte Carlo replicates is $7.8\%$ and increases to $12.53\%$  for the better statistic in \eqref{eqn:vgparallel_spectrum}. 
When  $\theta$ is estimated via MLE, the power increases to $10.16\%$ for \eqref{eqn:spectral} and to $13.49\%$ for  \eqref{eqn:vgparallel_spectrum}. This tells us that, when the information on the frequencies $\{\nut\}_{k=1}^K$ is available and is incorporated into the analysis through the use of $\gparallel$ and  MLE estimation, it almost doubles the power of the classical spectral statistic.
}

\section{Goodness-of-fit via divisible statistics}
\subsection{On the construction of  goodness-of-fit tests based on partial sums}
\label{sec:recovering_GOF}
Despite the advantages given by parameter estimation against alternatives of the form in \eqref{eqn:hhat} and the use of $\gtheta^{\parallela}$  in \eqref{eqn:gparallel} to increase the power of single divisible statistics, neither of these phenomena guarantees their adequacy for goodness-of-fit. 
Nevertheless,  power can be restored by using a family of divisible statistics.

As shown by \citet{khm84} for the case of simple hypotheses, one possible choice is the process of partial sums 
\begin{equation}
    \label{eqn:partialsums}
    \vtheta(\gtheta \mathds{1}_A)= \int_A\hspace{-0.05cm}\int \gthetaxy\vtheta(dx,dz)=\frac{1}{\sqrt{K}}\sum_{x_k\in A} \gthetak,
\end{equation}
in which  $\gtheta \mathds{1}_A$  denotes functions of the form $\gthetaxy\mathds{1}_{\{x\in A\}}$. 
When testing parametric hypotheses, the process of interest is
\begin{equation}
    \label{eqn:partialsums_theta}
    \vthetahat(\gthetahat \mathds{1}_A)= \vtheta\bigl(\Pi\gtheta\mathds{1}_A\bigl)+o_{_P}(1)
\end{equation}
in which the functions $\Pi\gtheta\mathds{1}_A$ can  be written explicitly as
\[\Pi\gthetaxy\mathds{1}_{\{x\in A\}}=\gthetaxy\mathds{1}_{\{x\in A\}}-\langle \gtheta\mathds{1}_A,\psitheta^{\sf T}\rangle\langle \btheta, \psitheta^{\sf T} \rangle^{-1} \bxy.\]

In this paper, we choose the sets $A$  to be members of the \emph{scanning family} \citep[cf.][]{khm93}, i.e., an increasing sequence of subsets, $\mathcal{A}=\{A_t\}_{0\leq t\leq 1}$, of $\X$ such that
\begin{enumerate}
    \item[(a)] $ A_s\subseteq A_{t} \text{ if } s\leq t$;
    \item[(b)] $ |A_0|=0, |A_1|=|\X|$;
    \item[(c)] $|A_t|$ is absolutely continuous in $t$.
    \end{enumerate}
It follows that one can label the functions in $\mathcal{L}_2(\mutheta)$  indexing the processes in \eqref{eqn:partialsums}-\eqref{eqn:partialsums_theta}  using the one-dimensional parameter $t\in[0,1]$.  As shown in Section \ref{sec:recovering_distr_free}, this also simplifies the construction of asymptotically distribution-free goodness-of-fit tests.

The process $\vtheta(\gtheta \mathds{1}_A)$ consists of partial sums of independent random variables; hence, it converges weakly to a Brownian motion \citep[cf.][]{pyke83,pyke86}. The projection 
$\vtheta\bigl(\Pi\gtheta\mathds{1}_A\bigl)$, on the other hand, converges weakly to a projected Brownian motion.  
Under the null,   the mean function of such processes is equal to zero for all $A\in \mathcal{A}$; whereas their shift under the alternatives can be derived by replacing  $\gtheta$ with $\gtheta\mathds{1}_A$ in \eqref{eqn:shift_g} and \eqref{eqn:shift_projection}. For instance,
\eqref{eqn:shift_g} becomes
 \begin{equation*}
         \label{eqn:shiftA}
\frac{1}{\sqrt{c}}\int_A C(x;\gtheta) h(x) \mu(dx).
     \end{equation*}
 Since the integral is taken over all $A\in \mathcal{A}$,    it is identically equal to zero only if $C(x;\gtheta)=0$ for all $x\in\X$ or if the functions  $h$ and $C(\dot;\gtheta)$ have exotic mutual alignment. Therefore, in practice, any choice of $\gtheta$ that differs from $\gtheta^{\perp}$ (cf. Section \ref{sec:linear}) would give a non-zero shift over at least some $A\in \mathcal{A}$.
Consequently, goodness-of-fit tests satisfying Definition \ref{def:gof_propetry} can be constructed by considering omnibus functionals of the partial sum processes in \eqref{eqn:partialsums}-\eqref{eqn:partialsums_theta} as test statistics. 

For example, let us start with a collection of hyperrectangles
 \begin{equation*}
\label{eqn:setsy}
(-\infty,y_t]=(-\infty,y_{1t}]\times\dots\times (-\infty,y_{dt}]
\end{equation*}
with $y_t=(y_{1t},\dots,y_{dt})\in\X$ such that, if $s\leq t$,  $y_{is}\leq y_{it}$, for all $i=1,\dots,d$. Then choose $A_t$ to be $(-\infty,y_t]\cap\X$. A possible test statistic for parametric goodness-of-fit is
\begin{equation}
    \label{eqn:KShat}
    \begin{split}
    \widehat{\text{KS}}=\max_{t}\bigl|\vthetahat (\gthetahat^{\parallela} \mathds{1}_t)\bigl| 
    \end{split}
\end{equation}
in which  $\mathds{1}_t$  denotes the indicator function    $\mathds{1}_{\{x\in (-\infty,y_t]\cap\X\}}$  and $\gtheta^{\parallela}$ is constructed as in \eqref{eqn:gparallel}.
The statistic in \eqref{eqn:KShat} is the Kolmogorov-Smirnov statistic for testing parametric hypotheses in the binned data regime. 

\emph{\textbf{Example IV (continued).}}
Consider the $\widehat{\text{KS}}$ statistics given by the maxima of the cumulative sums of the summands of the divisible statistics in \eqref{eqn:linear_stat} and \eqref{eqn:vgparallel_spectrum}, i.e.,
\begin{equation}
\label{eqn:KSss1}
\max_{t}\biggl|\frac{1}{\sqrt{K}}\sum_{x_k\leq y_t} \frac{(\nut-\mthetahat)}{\mthetahat}\biggl|
\end{equation}
and
\begin{equation*}
\label{eqn:KSss2}
\max_{t}\biggl|\frac{1}{\sqrt{K}}\sum_{x_k\leq y_t} p(q-1|\mthetahat)(\nut-\mthetahat)\biggl|
\end{equation*}
with $\{x_1,\dots,x_K\}$ linearly ordered, $y_t\in \{x_1,\dots,x_K\}$, and  $\theta$ estimated via MLE. 
Under the same setup described in Section \ref{sec:linear}, such statistics have power 
$16.47\%$ and $15.44\%$, respectively -- the power increases compared to  \eqref{eqn:linear_stat} and \eqref{eqn:vgparallel_spectrum}. 
 However, we should not expect this result to be true in general. Specifically, a test based on functionals of the partial sums process should be preferred to a test based on a single divisible statistic if we want a single test that can detect many different alternatives. Nevertheless, it is possible that, for some alternatives, a single divisible statistic may exhibit higher power than functionals of the partial sums process. 


\subsubsection{On the detection of non-constant background shapes in $X$-ray source-free spectra}
\label{sec:chandraII}
While in the analysis of $X$-ray spectra, \sara{weak} departures from uniformity cannot be detected by \sara{$C$-homogeneous} statistics (cf. Section \ref{sec:spacehomogeneous}),  Section \ref{sec:recovering_GOF} suggests that power can be retrieved through partial sums.

Considering the source-free Chandra spectrum described in Section \ref{sec:chandra_intro}, we tested the hypothesis of constant background using the $\widehat{KS}$ statistic constructed as in \eqref{eqn:KSss1} with $\mthetahat=\widehat{c}$, for all $k=1,\dots,K$. A parametric bootstrap simulation with 100,000 replicates yielded a p-value of $0.008$. Hence,  the uniform background assumption is rejected.

\begin{figure}[htb]
\centering
\makebox{   \includegraphics[width=66mm]{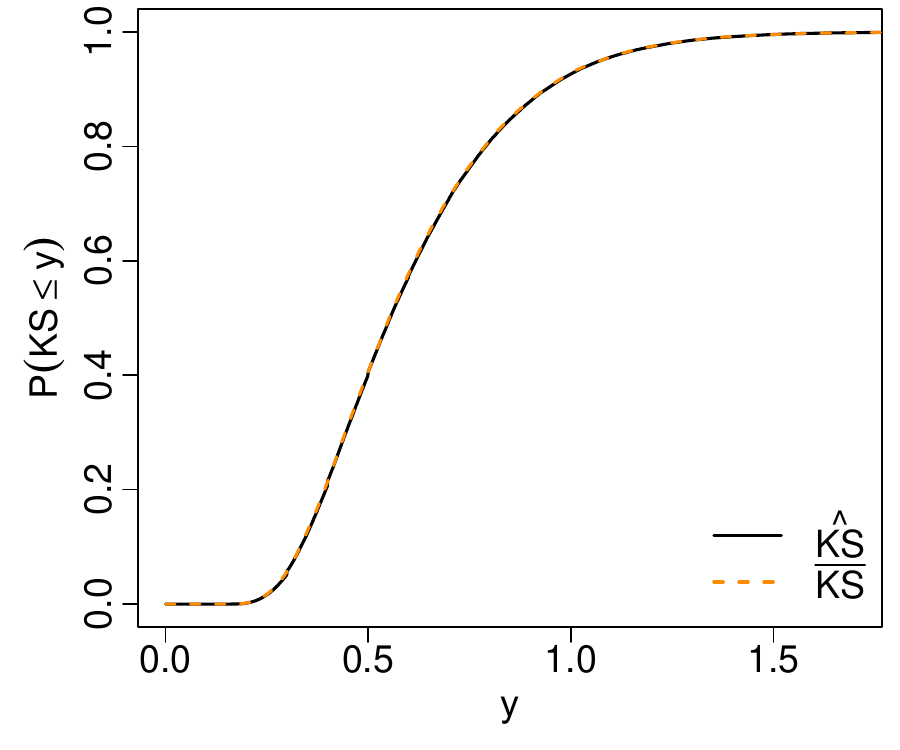} \includegraphics[width=100mm]{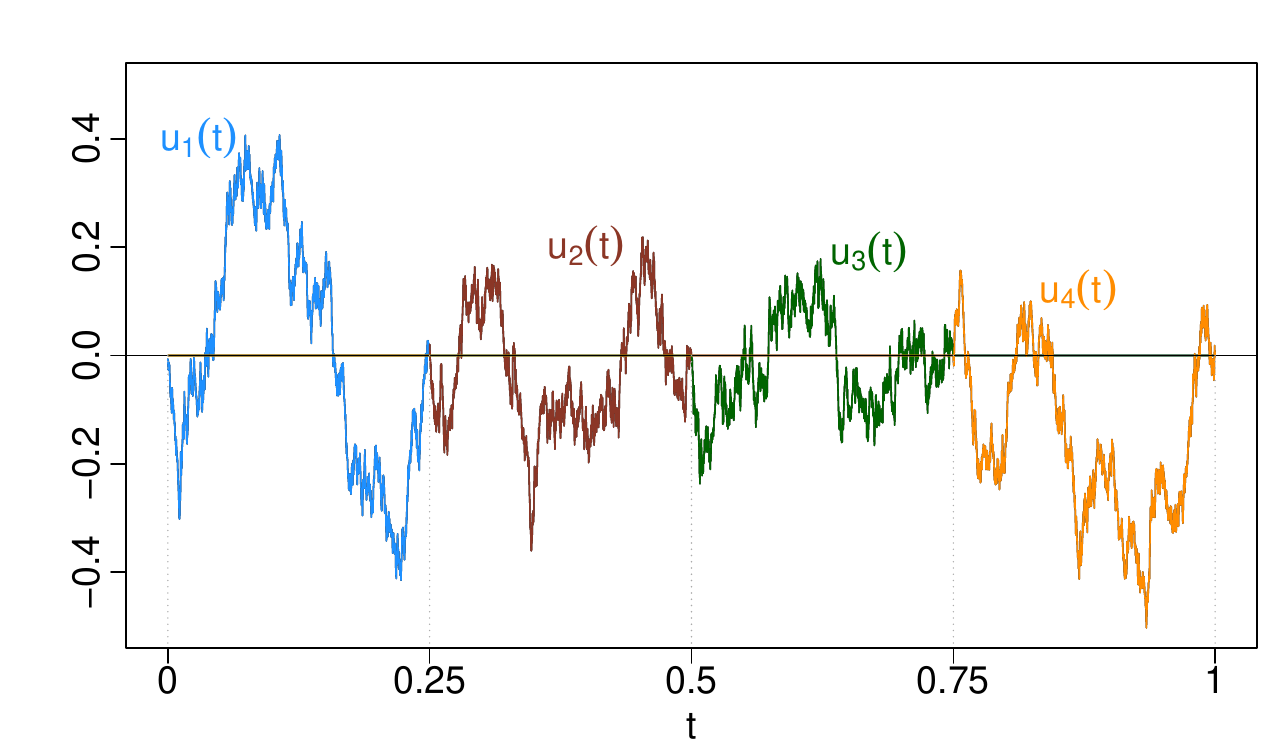}} 
\caption[Fig. 3]{
Left Panel: Comparing the bootstrapped distribution of the Kolmogorov-Smirnov statistics $\widehat{\text{KS}}$ and  $\overbar{\text{KS}}$, as defined in \eqref{eqn:KShat} and \eqref{eqn:KStilde}, using $100,000$ replicates. Right panel: A realization of the limiting process of  $\vtheta (U_p\Pi_r \ellthetat)$ given by a succession of rescaled independent standard Brownian bridges, defined as in \eqref{bridges}, when $p=4$. }
\label{fig3}
\end{figure}
The same testing procedure has been implemented to test the linearity of $m_{_\theta}$. The linear mean function estimated via MLE is displayed as a red-dashed line in Figure \ref{fig:RTCru_spectrum}. The bootstrap p-value obtained when testing such a model is $0.049$, on the edge of being rejected at a $5\%$ significance level.

While identifying a suitable model for this spectrum is beyond the scope of this manuscript, we tested a piecewise linear model comprising a linear function up to $15.6\r{A}$  and remaining constant thereafter.  The resulting fit is shown as a green-chained line in Figure \ref{fig:RTCru_spectrum}. When considering the $\widehat{KS}$ statistic in \eqref{eqn:KSss1}, the bootstrap p-value obtained amounts to $0.43$. This indicates that introducing a breakpoint at $15.6 \r{A}$  to separate the constant and linear components may provide a better representation of the background distribution in the $14.6-17.4 \r{A}$  range than models previously proposed in the literature.

\subsection{The projected parametric bootstrap}
\label{sec:bootstrap}
While the construction of asymptotically distribution-free tests for parametric hypotheses is possible  (cf. Section \ref{sec:recovering_distr_free}), one can alternatively derive the null distribution of functionals of $\vthetahat(\gthetahat\mathds{1}_A)$   via the parametric bootstrap. \sara{In standard settings, the consistency of the latter under \eqref{eqn:A2} has been established by several authors \citep[e.g.,][]{simonoff,sauermann}. Given the regularity of our framework, here, we are less concerned with the consistency of the parametric bootstrap; rather, we focus on demonstrating that the computational effort required by such a simulation procedure can be substantially reduced when using the projected process of partial sums. }

Let us focus on the situation in which the process of partial sums is indexed by functions $\gtheta^{\parallela}\mathds{1}_A$, and the parameters are estimated via MLE. Similar considerations can be made for any other choice of $\gtheta$ and other estimators solving the estimating equations in \eqref{eqn:general_est_eq}.

Denote with $\thetahat_{\text{obs}}$ the estimate of $\theta$ obtained via MLE on the set of data collected by the experiment. We are interested in simulating the null distribution of the process 
\begin{equation}
\label{eqn:main_process}
\vthetahat(\gthetahat^{\parallela}\mathds{1}_A)=\frac{1}{\sqrt{K}}\sum_{x_k\in A}\frac{C(x_k;\gthetahat)}{\mthetahat}(\nut-\mthetahat)
\end{equation}
and its functionals. 
To account for the randomness associated with  $\thetahat$, the classical parametric bootstrap entails maximizing the loglikelihood on each bootstrap sample and re-evaluating $\vthetahat(\gthetahat^{\parallela}\mathds{1}_A)$ at each bootstrap estimate of $\theta$.

Let us now consider the projected process of partial sums  
\begin{equation*}
\label{eqn:vgparallel}
\vtheta(\Pi\gtheta^{\parallela}\mathds{1}_A)=\vtheta(\gtheta^{\parallela}\mathds{1}_A) - \sum_{j=1}^p\langle \gtheta^{\parallela}\mathds{1}_A,s_j\rangle \vtheta(s_j).
\end{equation*}
Such a process is asymptotically equal to 
$\vthetahat(\gthetahat^{\parallela}\mathds{1}_A)$  (cf. Section \ref{sec:estimators}), however, from a computational standpoint, simulating the null distribution of $\vtheta(\Pi\gtheta^{\parallela}\mathds{1}_A)$  is more advantageous than simulating that of  $\vthetahat(\gthetahat^{\parallela}\mathds{1}_A)$. That is because, in the former case, one needs not to re-estimate $\theta$ on each bootstrap sample. Instead, the randomness induced by parameter estimation is accounted for through the second term on the right-hand side of the above expression in which the inner products $\langle \gtheta^{\parallela}\mathds{1}_A,s_j\rangle$ need only to be computed once at $\theta=\thetahat_{\text{obs}}$. 

\textbf{\emph{Numerical study.}} The computational benefits of the projected empirical process, particularly in the context of testing multivariate distributions in the i.i.d. setting, have been discussed in \cite{algeri22}.
To illustrate this aspect in the binned data regime, we \sara{use} the setup of Example IV introduced in Section \ref{sec:linear}. We compare the bootstrapped null distribution of the Kolmogorov-Smirnov statistics in \eqref{eqn:KShat} with that of 
\begin{equation}
\label{eqn:KStilde}
\overbar{\text{KS}}=\max_{t}\bigl|\vtheta(\Pi \gtheta^{\parallela}\mathds{1}_t)\bigl|,
\end{equation}
where $t$ labels the sets $(-\infty,y_t]\cap\X$ (cf. Section \ref{sec:recovering_GOF}) and $\gtheta^{\parallela}$ is the function defining the statistic in \eqref{eqn:linear_stat}.
The results obtained are shown on the left panel of Fig. \ref{fig3}.  
Not surprisingly, the two bootstrapped distributions are effectively overlapping.  The projected process $\vtheta(\Pi\gtheta^{\parallela}\mathds{1}_t)$,  however, provides substantial computational gain compared to $\vthetahat(\gthetahat^{\parallela}\mathds{1}_t)$. Specifically,  simulating the null distribution of $\overbar{\text{KS}}$ using $100,000$ bootstrap replicates required only 5.51 seconds of (system+user)  CPU time;  approximately 3.3  times faster than simulating that of  $\widehat{\text{KS}}$ and which required 18.38 seconds of CPU \sara{$C$-homogeneous}.

\subsection{Asymptotically distribution-free goodness-of-fit tests via unitary operators}
\label{sec:recovering_distr_free}
The projected bootstrap can reduce the CPU time needed to simulate the null distribution of goodness-of-fit statistics.
At the same time, it is also possible to construct goodness-of-fit tests based on the asymptotically distribution-free transformation of empirical processes introduced in  \citet{khm16}. If we do this, the limiting null distribution of the test statistics based on the transformed process will be known analytically and will not depend on the model we test. 
Such a feature is convenient for testing many different models simultaneously or when analyzing a large amount of data. For example, $X$-ray astronomical missions survey millions of astronomical sources. When analyzing the images and spectra arising from such observations, a common problem is identifying those in which features other than the background are present. In statistical terms, this problem involves testing spectral models on several millions of datasets, which would be unfeasible through case-by-case numerical simulations.
 
Given the prominence of the linear statistics, we consider the  functions
\begin{align*}
\ellthetat(x,z)= \frac{z-\mthetax}{\sqrt{\mthetax}}\mathds{1}_{\{x\in A_t\}}
\end{align*}
in which the sets $A_t$ are members of the scanning family $\mathcal{A}$ and are  chosen so that 
\begin{equation}
\label{eqn:mass}
\mu(A_t)=t.
\end{equation}
The last condition is not necessary but simplifies the notation in what follows. 
Under \eqref{eqn:A2} and if $\theta$ is known, the null distribution of the process $\vtheta(\ellthetat)$ converges weakly to that of a standard Brownian motion. When $\theta$ is estimated via MLE, $\vthetahat(\ellthetat)$ is asymptotically equivalent to $\vtheta(\Pi\ellthetat)$ and its limiting null distribution is that of a projected Brownian motion orthogonal to the score function $s$ of the model we wish to test. Thus, the null distribution of test statistics based on $\vthetahat(\ellthetat)$ will be different for different models and for different values of $\theta$ within the same model.

However, it is possible to construct a linear operator in $\mathcal{L}_2(\mutheta)$, denoted by $U_p$, which maps the process $\vtheta(\Pi\ellthetat)$ into a process in $t$ with standard distribution, unconnected with the model we are testing: the operator $U_p$ will depend on the model, but the distribution of the resulting process will be free from it.  In its spirit, the situation is not unlike the one encountered in the i.i.d. data regime in which the classical \sara{probability integral transform} $t = F(x)$ depends on the hypothetical $F$ but maps $F$-empirical processes into the uniform empirical process, which is free from $F$.

The idea behind the choice of the operator $U_p$ is geometrically clear: many different parametric models with a   $p$-dimensional estimated parameter will asymptotically lead to projected Brownian motions parallel to the corresponding score functions. Therefore, the kernels of these projections, albeit different for different models, will all have the same dimension $p$. But then one can choose one such projection orthogonal to a fixed collection of $p$ orthonormal functions and map particular projections corresponding to particular models into this chosen standard projection.

 Let $r= \{r_j\}_{j=1}^p$ form a fixed orthonormal system in $\mathcal{L}_2(\mutheta)$. Unlike the coordinates of the orthonormalized score function $s=\{s_j\}_{j=1}^p$, this new system will stay the same for many different models.
 Consider the projection of $\ellthetat$ orthogonal to $r$:
$$\Pi_r \ellthetat(x,z)=\ellthetat(x,z) - \sum_{j=1}^p \langle \ellthetat,r_j\rangle r_j(x,z) .$$
Choose $U_p$ to be the unitary operator that maps $r$ into $s$ and, therefore, functions orthogonal to $r$ into functions orthogonal to $s$. Then, the process $\vthetahat (U_p\Pi_r \ellthetat)$ serves our purpose since
$$\vthetahat (U_p\Pi_r \ellthetat) \sim \vtheta (\Pi U_p\Pi_r \ellthetat) = \vtheta (U_p\Pi_r \ellthetat)$$
and the variance of the process on the right-hand side -- and, therefore, its covariance -- depends on $r$ but does not involve $s$:
\begin{equation}\label{two}
\Etheta[v^2_{_{\theta,K}} (U_p\Pi_r \ellthetat) ]  =\| U_p\Pi_r \ellthetat \|^2  = \| \Pi_r \ellthetat \|^2  = \| \ellthetat \|^2  - \sum_{j=1}^p \langle \ellthetat, r_j\rangle^2.
\end{equation}
Thus, the limiting distribution of statistics based on the transformed process $v_{_{\theta,K}}(U_p\Pi_r  \ellthetat)$ will need to be found only once and then used for testing many different models.

\sara{
Furthermore, as formalized in Proposition \ref{prop:collection_bridges}, we can suggest a particular choice of $r$ that allows us to derive the asymptotic null distribution of goodness-of-fit statistics in closed form.
\begin{proposition}
\label{prop:collection_bridges}
Consider a sequence $\{A_{t_j}\}_{j=1}^p$  of sets within the scanning family $\mathcal{A}$ such that $\mu(A_{t_j})=j/p$ and denote  with  $\{B_j\}_{j=1}^p$ the collection of increments
$B_j=A_{t_j}\setminus A_{t_{j-1}}$,
with $B_1=\X$ if $p=1$. 
Choose $r= \{r_j\}_{j=1}^p$ so that:
\begin{equation}
\label{eqn:rj}
\phithetaj (x,z)=\frac{z-\mthetax}{\sqrt{\mthetax}}\sqrt{p}\mathds{1}_{\{x\in B_j\}}.
\end{equation}
Then, if $H_0$ holds, the process  $\vtheta(U_p\Pi_r\ellthetat)$ converges in distribution to a sequence of independent standard Brownian bridges rescaled over the sets $B_j$, i.e., 
\begin{equation}
\label{bridges}
u_j (t) = \frac{1}{\sqrt{p}} u(p(t - t_{j-1})), \; \; t_{j - 1} < t \le t_j,
\end{equation}
with $u(t),t\in[0,1]$ denoting the standard Brownian bridge.
\end{proposition}
\begin{proof}
Since each $U_p\Pi_r\ellthetat$ is centered under $H_0$, $\vtheta(U_p\Pi_r\ellthetat)$ has zero mean. Moreover, split $\Pi_r \ell_t$ as the sum
$$\Pi_r \ell_t = \sum_{j=1}^p \ell_t{\mathds 1}_{\{x\in B_j\}} - \sum_{j=1}^p \langle r_j, \ell_t \rangle r_j = \sum_{j=1}^p \left[\ell_t{\mathds 1}_{\{x\in B_j\}} - \langle r_j, \ell_t \rangle r_j \right ] .$$
In this sum, only the summand for which $t_{j-1}<t\le t_j$ is non-zero. For this summand
$$\|\ell_t{\mathds 1}_{\{x\in B_j\}}\|^2\to t- t_{j-1}, \; \; \text{while} \; \; \; \langle r_j, \ell_t \rangle^2 \to p(t-t_{j-1})^2 $$
and eventually, we obtain
$$ \Etheta [v^2_{_{\theta,K}} (U_p\Pi_r \ellthetat)] \to t- t_{j-1} - p(t-t_{j-1})^2, \;\; \; t_{j - 1} < t \le t_j.$$
When comparing the last expression with the variance of $u_j (t)$,
we find that they are the same. 
\qed
\end{proof}}
\begin{figure}
\centering
\makebox{\includegraphics[width=55mm]{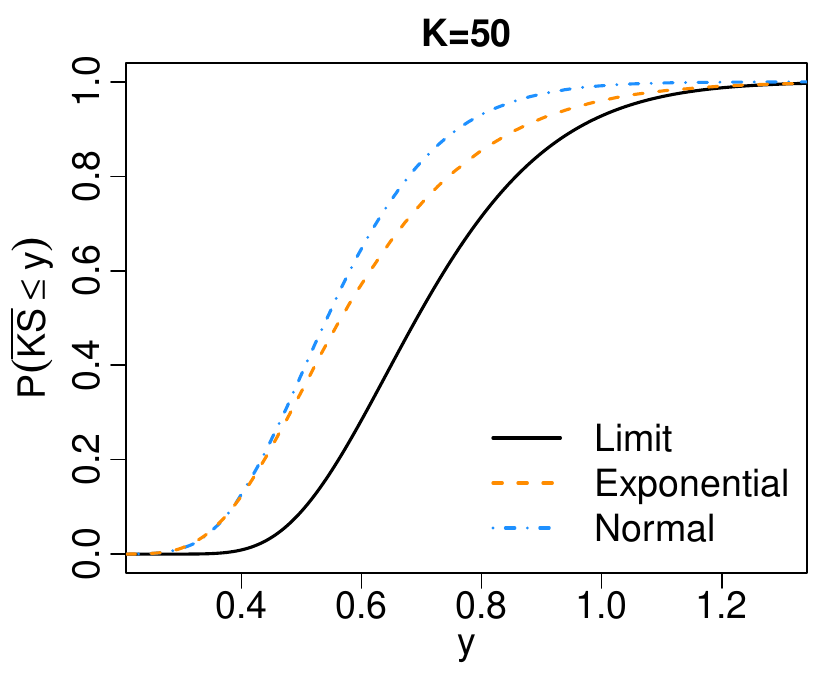}  \includegraphics[width=55mm]{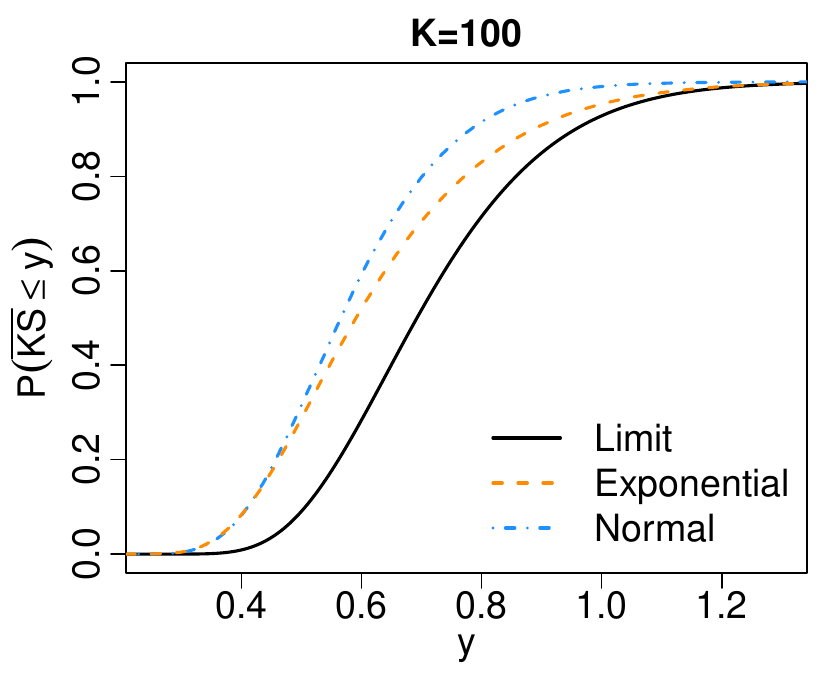}\includegraphics[width=55mm]{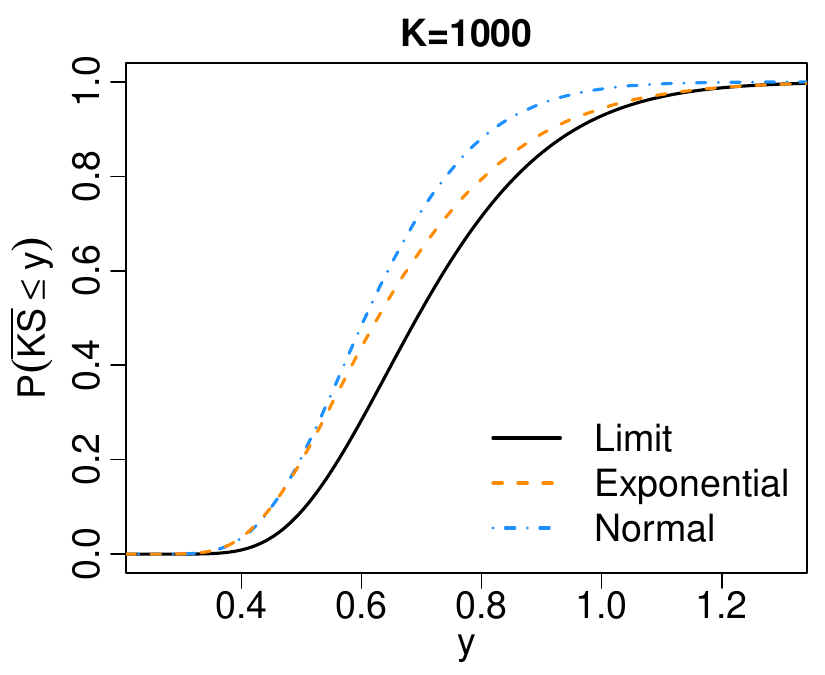}}\\
\makebox{\includegraphics[width=55mm]{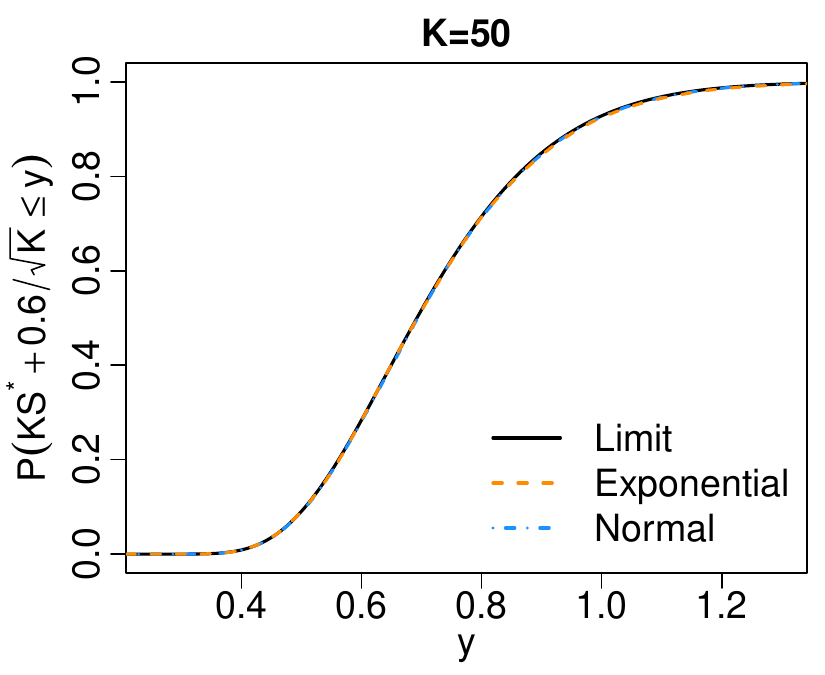}  \includegraphics[width=55mm]{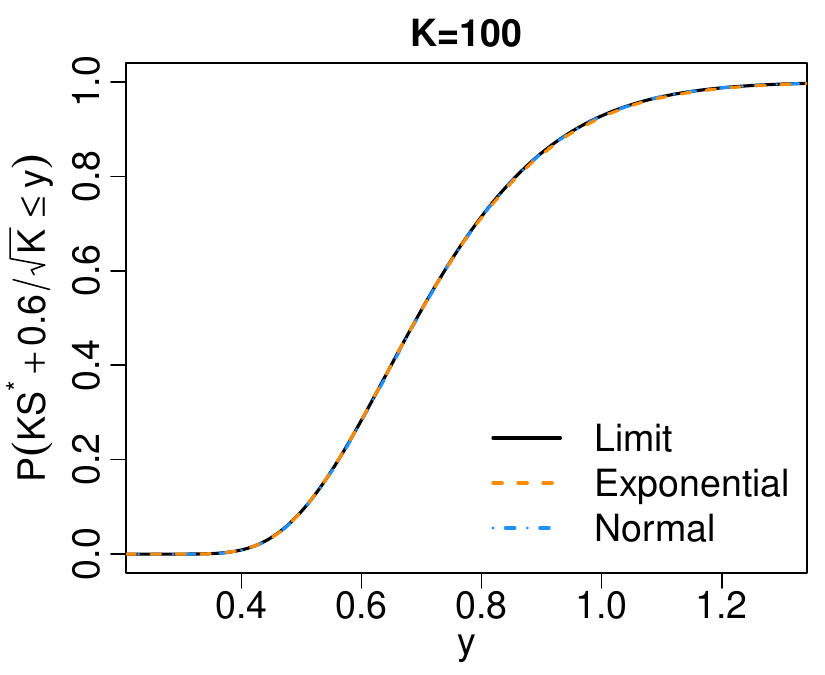}\includegraphics[width=55mm]{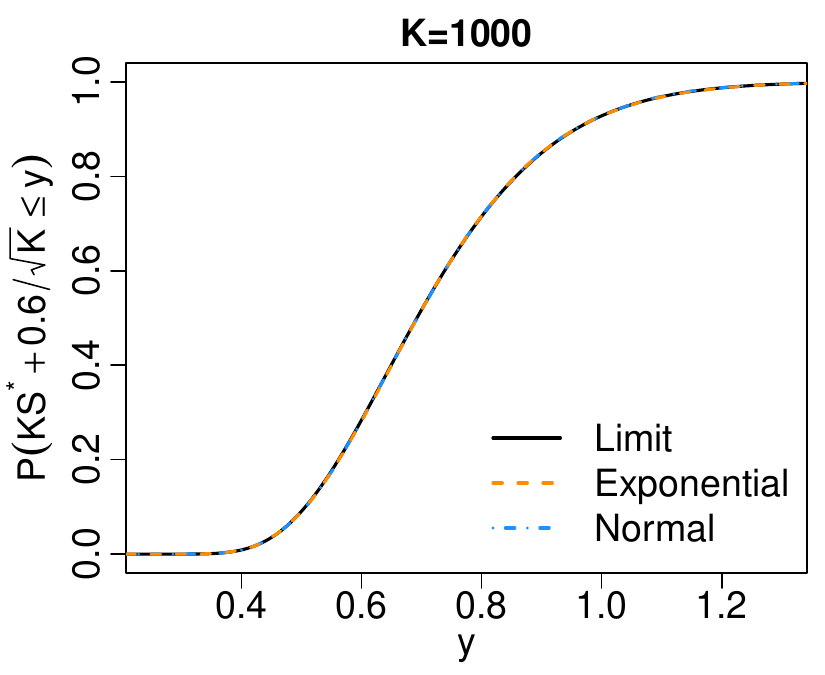}}\\
\caption[Fig. 5]{
Simulated null distributions of the test statistic $\overbar{\text{KS}}$ in \eqref{eqn:KStilde} (top panels) 
and its transformed counterpart $\text{KS}^*$  in \eqref{eqn:KS_star} shifted by $\frac{0.6}{\sqrt{K}}$ (bottom panels), for the exponential (orange dashed lines) and the normal (blue chained lines) models for $K=50,100,1000$. For comparison, the limiting distribution in \eqref{eqn:kolmogorov} is also plotted in both panels (black solid lines).  Each simulation has been conducted using $100,000$ replicates. }
\label{fig5}
\end{figure}

\sara{The right panel of Fig. \ref{fig3} shows one realization of $\vtheta(U_p\Pi_r\ellthetat)$ when $r$ is chosen as in \eqref{eqn:rj} and $p=4$. }
\sara{When $\vtheta(U_p\Pi_r\ellthetat)$ is constructed as in Proposition \ref{prop:collection_bridges}, the asymptotic null distribution of many goodness-of-fit statistics from it is analytically known.} For example, consider 
\begin{equation}
\label{eqn:KS_star}
\text{KS}^*=\max_{t}\bigl|\vtheta( U_p\Pi_r \ellthetat)\bigl|.
\end{equation} 
Since the maximum of a standard Brownian bridge is known to follow the Kolmogorov distribution \citep{kolmogorov}, here denoted with $\mathcal{K}(\cdot)$, we have that, as $K\rightarrow\infty$,
\begin{equation}
\label{eqn:kolmogorov}
P\bigl(\text{KS}^*\leq y|H_0\bigl)\rightarrow P\Bigl(\max_{i}\bigl\{\max_{t}\bigl|\sqrt{p}u_i(t)\bigl|\bigl\}\leq \sqrt{p}y\Bigl)=\Bigl[\mathcal{K}(\sqrt{p}y)\Bigl]^p.
\end{equation}

 The partial sums process, however, is observed over a discretized space; hence, its maximum is smaller than or equal to the maximum of the Brownian bridge, leading to a shift in the distribution. The application of the continuity correction
\begin{equation*}
\label{eqn:kolmogorov_adj}
\Bigl[\mathcal{K}\bigl(\sqrt{p}(y+{\textstyle\frac{0.6}{\sqrt{K}}})\bigl)\Bigl]^p.
\end{equation*}
agrees with the finite-$K$ distribution even for $K\sim 10$. 

It remains to address how to construct the unitary operator $U_p$.
Consider the unitary operator $U_{a,b}$ defined as 
\begin{equation*}
\label{eqn:U}
U_{a,b}=I-\frac{\langle\cdot , a-b\rangle}{1-\langle a,b\rangle} [a(x,z)-b(x,z)]
\end{equation*}
where $a,b\in \mathcal{L}_2(\mutheta)$ and $\|a\|=\|b\|=1$. $U_{a,b}$ maps $a$ into $b$, $b$ into $a$, and leaves functions orthogonal to both $a$ and $b$ unchanged.  
Proceed by constructing a set of auxiliary functions $\phithetatildej$ such that $\widetilde{r}_1={r}_1$ and, for $j>1$,  set
\begin{equation}
\label{eqn:Uj}
\phithetatildej(x,z)=U_{j-1}r_{j}(x,z)\quad\text{with}\quad U_{j}=U_{\widetilde{r}_j,s_j}\dots U_{\widetilde{r}_1,s_1}.\end{equation}
For example, 
\begin{equation*}
\widetilde{r}_{2}\hspace{-0.02cm}(x,z)=U_{\widetilde{r}_1,s_1} r_2\hspace{-0.02cm}(x,z), \quad
\widetilde{r}_{3}\hspace{-0.02cm}(x,z)=U_{\widetilde{r}_2,s_2}U_{\widetilde{r}_1,s_1}r_3\hspace{-0.02cm}(x,z), \quad \text{etc.}
\end{equation*}
Each function $\phithetatildej$ is orthogonal to all $s_{i}$ with $i<j$.

The operator $U_p$, constructed as in \eqref{eqn:Uj},
is a product of unitary operators; hence, it is unitary, and it can be shown that $U_pr_j=s_j$ for all $j=1,\dots,p$.
 Furthermore, given that $\Pi_r$ and $U_p$ are linear operators,  $\vtheta( U_p\Pi_r\ellthetat)$ is a weighted linear statistic for each $t$ fixed; thus, the corresponding  $\vtheta(U_p\Pi_r\ellthetat^\perp)$ component (cf. Section \ref{sec:linear}) is always zero.

\textbf{\emph{Numerical study.}} 
To validate the performance of the approach described in this section in retrieving distribution-freeness, 
  consider the exponential and normal models in Examples I and IV. In both instances, assume $\theta=(c,\beta)^{\sf T}$ to be unknown. Furthermore, to assess how large $K$ needs to be for our asymptotic result to hold, we consider $K=50,100,1000$.  The top panels of Fig. \ref{fig5} show the simulated null distributions of the statistic $\overbar{\text{KS}}$ in \eqref{eqn:KStilde} with $\gparallel\mathds{1}_t=\ellthetat$ and $t\in\{{\scriptstyle 0,\frac{1}{K},\frac{2}{K},\dots,1\}}$.  Due to the absence of distribution-freeness, regardless of how large $K$ is, these null distributions differ substantially when testing the exponential model (orange dashed line) and the normal model (blue chained line). As expected,   they also deviate from the limiting distribution in \eqref{eqn:kolmogorov} (black solid lines).
Conversely, the bottom panels of Fig. \ref{fig5} show that, for all the values of $K$ considered, 
the simulated null distributions of the $\text{KS}^*$ in \eqref{eqn:KS_star}, adjusted by a factor of $\frac{0.6}{\sqrt{K}}$, is the same when testing either the exponential or the normal model. Moreover, both distributions practically coincide with the limit in \eqref{eqn:kolmogorov}.

\section{Summary of the main results and final remarks}
\label{sec:discussion}
This article demonstrates that various divisible statistics, of different flavour and structure, can be redefined as linear functionals of the same empirical process (Section \ref{sec:divisible}), thereby extending the class of divisible statistics and making their asymptotic theory transparent and simple.

The proposed study reveals phenomena with no analogues in the classical theory of empirical processes. 
In particular, the discovery of the existence of the function $C(\cdot;\gtheta)$, which defines the shift of divisible statistics under alternatives (Propositions \ref{prop:magic_form}-\ref{prop:powerMLE}(i)), enabled us to show how to construct, for any divisible statistic, a ``better'' divisible statistic, with higher power. These ``better'' statistics are members of the class of weighted linear statistics (Proposition \ref{prop:linear}). Moreover, although divisible statistics for which  $C(\cdot;\gtheta)$ is constant are commonly employed in many practical applications,  they are unreliable in testing problems (Section \ref{sec:spacehomogeneous}). This fact has also been illustrated through the data from the Chandra $X$-ray Observatory.

For parameter estimation based on weighted divisible statistics, the optimal weighting function was identified (Proposition \ref{prop:optimality}) -- a result of Cram\'er-Rao type. Furthermore, it was confirmed that test statistics with estimated parameters can yield higher power than testing simple hypotheses (Proposition \ref{prop:powerMLE}(ii)). 
 
To derive the distribution of goodness-of-fit statistics based on omnibus functionals of partial sums of divisible statistics, two approaches are offered: a fast bootstrap procedure (Section \ref{sec:bootstrap}) and the construction of distribution-free tests with a known closed-form distribution (Section \ref{sec:recovering_distr_free}).

The theoretical framework introduced  immediately applies to Poisson regression problems, and it naturally lends itself to any other situation in which non-identically distributed data are encountered. It would be interesting to explore extensions of the proposed setup to high-dimensional regression in which the dimension of the parameter space is much larger than the sample size. Such a development could broaden further the set of inferential tools currently available in this context \citep[e.g.,][]{chen2001weak,VandeGeer,jankova}.

\section*{Supplementary Material}
Technical details are available in the Supplementary Material.

\section*{Data availability}
The \texttt{R} code and the data needed to replicate the results presented in this manuscript are available at \url{https://github.com/salgeri/Algeri\_Khmaladze\_grouped_data}

\section*{Acknowledgments}
S.A. is grateful to the School of Mathematics and Statistics at the Victoria University of Wellington for providing the resources and fostering a welcoming environment at the beginning of this project. She also thanks Bob Cousins,     Vinay Kashyap, Knut  Mor{\aa}, and Lawrence Rudnick for the useful discussions on goodness-of-fit problems arising in physics and astronomy. Both authors thank three anonymous referees and the editor for the thoughtful comments and constructive feedback. 
\section*{Funding}
This research is partially supported by the Office of the Vice President for Research at the University of Minnesota,  by the NSF grant DMS-2152746, and the Marsden grant VUW1616.

\bibliographystyle{abbrvnat}
\bibliography{biblio}

@book{rudin87,
  title={Real and complex analysis},
  author={Rudin, W.},
  year={1987},
  publisher={McGraw-Hill}
}

@book{lehmann1986testing,
  title={Testing statistical hypotheses},
  author={Lehmann, E.L. and Romano, J.P. },
  volume={3},
  year={1986},
  publisher={Springer}
}

@article{domanski,
  title={Statistical tests based on empty cells},
  author={Domanski, C.},
  journal={Acta Universitatis Lodziensis},
  number={269},
  year={2012}
}

@book{baayen,
  title={Word frequency distributions},
  author={Baayen, H},
  year={2001},
  publisher={Kluwer}
}

@book{magurran,
  title={Ecological Diversity and its Measurement},
  author={Magurran, A.E.},
  year={1988},
  publisher={Princeton University Press}
}

@article{sauermann,
  title={Bootstrapping the maximum likelihood estimator in high-dimensional log-linear models},
  author={Sauermann, W},
  journal={The Annals of Statistics},
  pages={1198--1216},
  year={1989},
  publisher={JSTOR}
}

@article{simonoff,
  title={Jackknifing and bootstrapping goodness-of-fit statistics in sparse multinomials},
  author={Simonoff, J.S.},
  journal={Journal of the American Statistical Association},
  volume={81},
  number={396},
  pages={1005--1011},
  year={1986},
  publisher={Taylor \& Francis}
}

@article{roberts,
  title={Distribution free testing for the family of Laplace distributions},
  author={Roberts, B. and Haywood, J.n and Swordson, E.},
  journal={Communications in Statistics-Simulation and Computation},
  pages={1--12},
  year={2023},
  publisher={Taylor \& Francis}
}

@article{cai,
  title={Testing high-dimensional multinomials with applications to text analysis},
  author={Cai, T. T. and Ke, Z.T. and Turner, P.},
  journal={Journal of the Royal Statistical Society Series B: Statistical Methodology},
  pages={qkae003},
  year={2024},
  publisher={Oxford University Press}
}

@article{jankova,
  title={Goodness-of-fit testing in high dimensional generalized linear models},
  author={Jankov{\'a}, J. and Shah, R.D. and B{\"u}hlmann, P. and Samworth, R.J.},
  journal={Journal of the Royal Statistical Society Series B: Statistical Methodology},
  volume={82},
  number={3},
  pages={773--795},
  year={2020},
  publisher={Oxford University Press}
}

@article{bonamente,
  title={Distribution of the C statistic with applications to the sample mean of Poisson data},
  author={Bonamente, M.},
  journal={Journal of Applied Statistics},
  volume={47},
  number={11},
  pages={2044--2065},
  year={2020},
  publisher={Taylor \& Francis}
}

@article{vinay,
  title={Long-term X-ray Variability of the Symbiotic System RT Cru based on Chandra Spectroscopy},
  author={Danehkar, A. and Karovska, M. and Drake, J.J. and Kashyap, V.L.},
  journal={Monthly Notices of the Royal Astronomical Society},
  year={2020}
}

@article{luna,
  title={The nature of the hard X-ray-emitting symbiotic star RT Cru},
  author={Luna, G.J.M. and Sokoloski, J.L.},
  journal={The Astrophysical Journal},
  volume={671},
  number={1},
  pages={741},
  year={2007},
  publisher={IOP Publishing}
}

@article{zhang,
  title={A novel approach to detect line emission under high background in high-resolution X-ray spectra},
  author={Zhang, X. and Algeri, S. and Kashyap, V. and Karovska, M.},
  journal={Monthly Notices of the Royal Astronomical Society},
  volume={521},
  number={1},
  pages={969--983},
  year={2023},
  publisher={Oxford University Press}
}

@article{dumbgen,
  title={A new approach to tests and confidence bands for distribution functions},
  author={D{\"u}mbgen, L. and Wellner, J.A.},
  journal={The Annals of Statistics},
  volume={51},
  number={1},
  pages={260--289},
  year={2023},
  publisher={Institute of Mathematical Statistics}
}

@article{smirnov37,
  title={On the distribution of the von {M}ises $\omega^2$-criterion},
  author={Smirnov, N.},
  journal={Matematicheskii Sbornik},
  volume={5},
  pages={973--993},
  year={1937}
}

@article{karovska2010,
  title={A Precessing Jet in the CH Cyg Symbiotic System},
  author={Karovska, M. and Gaetz, T.J. and Carilli, C.L. and Hack, W. and Raymond, J.C. and Lee, N.P.},
  journal={The Astrophysical Journal Letters},
  volume={710},
  number={2},
  pages={L132},
  year={2010},
  publisher={IOP Publishing}
}

@article{belforte,
  title={Search for resonant and nonresonant new phenomena in high-mass dilepton final states at {$\sqrt{s}= 13$} {TeV}},
  author={Belfonte, S. and others},
  journal={Journal of High Energy Physics},
  volume={2021},
  number={07},
  pages={1--62},
  year={2021}
}

@article{khm2007,
  title={Differentiation of sets in measure},
  author={Khmaladze, E.V.},
  journal={Journal of Mathematical Analysis and Applications},
  volume={334},
  number={2},
  pages={1055--1072},
  year={2007},
  publisher={Elsevier}
}

@article{khm2008,
  title={Local empirical processes near boundaries of convex bodies},
  author={Khmaladze, E. and Weil, W.},
  journal={Annals of the Institute of Statistical Mathematics},
  volume={60},
  number={4},
  pages={813--842},
  year={2008},
  publisher={Springer}
}

@article{broken_pl1,
  title={The ultraluminous state},
  author={Gladstone, J. C. and Roberts, T.P. and Done, C.},
  journal={Monthly Notices of the Royal Astronomical Society},
  volume={397},
  number={4},
  pages={1836--1851},
  year={2009},
  publisher={Blackwell Publishing Ltd Oxford, UK}
}

@article{broken_pl2,
  title={A {C}handra view of the normal {S0} galaxy {NGC} 1332. {I}. An unbroken, steep power-law luminosity function for the low-mass {X}-ray binary population},
  author={Humphrey, P.J. and Buote, D.A.},
  journal={The Astrophysical Journal},
  volume={612},
  number={2},
  pages={848},
  year={2004},
  publisher={IOP Publishing}
}

@article{broken_pl3,
  title={Spectral state transitions of the Ultraluminous {X}-ray Source {IC} 342 {X}-1},
  author={Marlowe, Hannah and Kaaret, Philip and Lang, Cornelia and Feng, Hua and Grise, Fabien and Miller, Neal and Cseh, David and Corbel, Stephane and Mushotzky, Richard F},
  journal={Monthly Notices of the Royal Astronomical Society},
  volume={444},
  number={1},
  pages={642--650},
  year={2014},
  publisher={Oxford University Press}
}

@book{cramer,
  title={Mathematical methods of statistics},
  author={Cram{\'e}r, H.},
  year={1999},
  publisher={Princeton University Press}
}

@article{smirnov,
  title={On the estimation of the discrepancy between empirical curves of distribution for two independent samples},
  author={Smirnov, N.V.},
  journal={Moscow University Mathematics Bulletin},
  volume={2},
  number={2},
  pages={3--14},
  year={1939}
}

@article{kolmogorov,
  title={Sulla determinazione empirica di una lgge di distribuzione},
  author={Kolmogorov, A.},
  journal={Giornale dell'Instituto Italiano degli Attuari},
  volume={4},
  pages={83--91},
  year={1933}
}

@article{khm16,
  title={Unitary transformations, empirical processes and distribution free testing},
  author={Khmaladze, E.},
  journal={Bernoulli},
  volume={22},
  number={1},
  pages={563--588},
  year={2016},
  publisher={Bernoulli Society for Mathematical Statistics and Probability}
}

@article{VandeGeer,
author = {van de Geer, S. and B{\"u}hlmann, P. and  Ritov, Y. and  Dezeure, R.},
  title = {{On asymptotically optimal confidence regions and tests for high-dimensional models}},
volume = {42},
journal = {The Annals of Statistics},
number = {3},
publisher = {Institute of Mathematical Statistics},
pages = {1166 -- 1202},
year = {2014}
}

@article{chen2001weak,
  title={Weak convergence of the empirical process of residuals in linear models with many parameters},
  author={Chen, G. and Lockhart, R.A.},
  journal={Annals of Statistics},
  pages={748--762},
  year={2001},
  publisher={JSTOR}
}

@article{henze,
  title={Recent and classical tests for exponentiality: a partial review with comparisons},
  author={Henze, N. and Meintanis, S.G.},
  journal={Metrika},
  volume={61},
  pages={29--45},
  year={2005},
  publisher={Springer}
}

@article{moscovich,
author = {Moscovich, A. and  Nadler, B. and  Spiegelman,C.},
title = {{On the exact Berk-Jones statistics and their $p$-value calculation}},
volume = {10},
journal = {Electronic Journal of Statistics},
number = {2},
publisher = {Institute of Mathematical Statistics and Bernoulli Society},
pages = {2329 -- 2354},
year = {2016},
doi = {10.1214/16-EJS1172}
}

@article{durio,
  title={Local efficiency of integrated goodness-of-fit tests under skew alternatives},
  author={Durio, Alessandra and Nikitin, Ya Yu},
  journal={Statistics \& Probability Letters},
  volume={117},
  pages={136--143},
  year={2016},
  publisher={Elsevier}
}

@article{anderson,
  title={A test of goodness of fit},
  author={Anderson, T.W. and Darling, D.A.},
  journal={Journal of the American Statistical Association},
  volume={49},
  number={268},
  pages={765--769},
  year={1954},
  publisher={Taylor \& Francis}
}

@book{vandervaart,
  title={Asymptotic Statistics},
  author={van der Vaart, A.W.},
  volume={3},
  year={2000},
  publisher={Cambridge University Press}
}

@article{barbour,
  title={Small counts in the infinite occupancy scheme},
  author={Barbour, A. and Gnedin, A.},
  journal={Electronic Journal of Probability},
  volume={14},
  pages={365--384},
  year={2009},
  publisher={Institute of Mathematical Statistics and Bernoulli Society}
}

@article{lecam,
  title={Locally asymptotically normal families of distributions},
  author={Le Cam, L.},
  journal={University of California Publications in Statistics},
  volume={3},
  pages={37--98},
  year={1960}
}

@article{GneHanPit07,
  title={Notes on the occupancy problem with infinitely many boxes: general asymptotics and power laws},
  author={Gnedin, A. and Hansen, B. and Pitman, J.},
journal={Probability Surveys},
volume = {4},
pages = {146 -- 171},
  year={2007}
}

@article{khm11,
  title={Convergence properties in certain occupancy problems including the {K}arlin-{R}ouault law},
  author={Khmaladze, E.V.},
  journal={Journal of Applied Probability},
  volume={48},
  number={4},
  pages={1095--1113},
  year={2011},
  publisher={Cambridge University Press}
}

@article{mnatsakanov86,
  title={A functional limit theorem for additively separable statistics in the case of very rare events},
  author={Mnatsakanov, R.M.},
  journal={Theory of Probability and Its Applications},
  volume={30},
  number={3},
  pages={622--626},
  year={1986},
  publisher={SIAM}
}

@article{mnatsakanov88,
  title={On the convergence of separable statistics to a {W}iener process},
  author={Mnatsakanov, R.M.},
  journal={Theory of Probability and Its Applications},
  volume={32},
  number={1},
  pages={152--157},
  year={1988},
  publisher={SIAM}
}

@article{medvedev70,
  title={Some theorems on the asymptotic distribution of the $\chi^2$ statistic},
  author={Medvedev, Yu.},
  journal={Doklady Akademii Nauk},
  volume={192},
  number={5},
  pages={987--989},
  year={1970},
  publisher={Russian Academy of Sciences}
}

@book{amari,
  title={Differential-Geometrical Methods in Statistics},
  author={Amari, S.-i.},
  year={1985},
  publisher={Lecture Notes in Statistics. Springer}
}

@article{mirakhmedov,
  title={Asymptotic normality associated with generalized occupancy problems},
  author={Mirakhmedov, S.M.},
  journal={Statistics and Probability Letters},
  volume={77},
  number={15},
  pages={1549--1558},
  year={2007},
  publisher={Elsevier}
}

@article{cressie_read,
  title={Multinomial goodness-of-fit tests},
  author={Cressie, N. and Read, T.R.C.},
  journal={Journal of the Royal Statistical Society: Series B (Methodological)},
  volume={46},
  number={3},
  pages={440--464},
  year={1984},
  publisher={Wiley Online Library}
}

@book{kass2011,
  title={Geometrical Foundations of Asymptotic Inference},
  author={Kass, R.E. and Vos, P.W.},
  year={2011},
  publisher={John Wiley \& Sons}
}

@article{khm93,
  title={Goodness of fit problem and scanning innovation martingales},
  author={Khmaladze, E.V.},
  journal={The Annals of Statistics},
volume={21},
number={2},
  pages={798--829},
  year={1993},
  publisher={JSTOR}
}

@article{oosterhoff,
  title={A note on contiguity and {H}ellinger distance},
  author={Oosterhoff, J. and Van Zwet, W.R.},
  journal={Reidel, Dordrecht },
  volume={157},
  pages={166},
  year={1979}
}

@article{chibisov65,
  title={An investigation of the asymptotic power of the tests of fit},
  author={Chibisov, D.M.},
  journal={Theory of Probability and Its Applications},
  volume={10},
  number={3},
  pages={421--437},
  year={1965},
  publisher={SIAM}
}

@article{ivchenko79,
  title={Separable statistics and hypothesis testing. The case of small samples},
  author={Ivchenko, G.I. and Medvedev, Yu.},
  journal={Theory of Probability and Its Applications},
  volume={23},
  number={4},
  pages={764--775},
  year={1979},
  publisher={SIAM}
}

@article{ivchenko81,
  title={Decomposable statistics and hypothesis testing for grouped data},
  author={Ivchenko, G.I. and Medvedev, Yu.},
  journal={Theory of Probability and Its Applications},
  volume={25},
  number={3},
  pages={540--551},
  year={1981},
  publisher={SIAM}
}

@article{khm98,
  title={Goodness of fit tests for “chimeric” alternatives},
  author={Khmaladze, E.V.},
  journal={Statistica Neerlandica},
  volume={52},
  number={1},
  pages={90--111},
  year={1998},
  publisher={Wiley Online Library}
}

@article{wellnerrev,
  title={Empirical processes in action: a review},
  author={Wellner, J.A.},
  journal={International Statistical Review},
volume={60},
  number={3},
  pages={247--269},
  year={1992},
  publisher={JSTOR}
}

@article{atwood,
  author={Atwood et al.,W. B. },
  title={The {L}arge {A}rea {T}elescope on the {F}ermi {G}amma-{R}ay {S}pace {T}elescope {M}ission},
  journal={The Astrophysical Journal},
  volume={697},
  number={2},
  pages={1071},
  year={2009}
}

@book{stigler,
  title={The History of Statistics},
  author={Stigler, S.M.},
  year={1990},
  publisher={Harvard University Press}
}

@article{cochran,
  title={The $\chi^2$ test of goodness of fit},
  author={Cochran, W.G.},
  journal={The Annals of Mathematical Statistics},
  volume={23},
  number={3},
  pages={315--345},
  year={1952},
  publisher={JSTOR}
}

@article{balakrishnan1,
  title={Hypothesis testing for densities and high-dimensional multinomials},
  author={Balakrishnan, S. and Wasserman, L.},
  journal={The Annals of Statistics},
  volume={47},
  number={4},
  pages={1893--1927},
  year={2019},
  publisher={JSTOR}
}

@article{balakrishnan2,
author = {Balakrishnan, S. and Wasserman, L.},
title = {{Hypothesis testing for high-dimensional multinomials: A selective review}},
volume = {12},
journal = {The Annals of Applied Statistics},
number = {2},
publisher = {Institute of Mathematical Statistics},
pages = {727 -- 749},
year = {2018}
}

@article{muller03,
  title={Asymptotic normality of goodness-of-fit statistics for sparse {P}oisson data},
  author={M{\"u}ller, U.U. and Osius, G.},
  journal={Statistics: A Journal of Theoretical and Applied Statistics},
  volume={37},
  number={2},
  pages={119--143},
  year={2003},
  publisher={Taylor \& Francis}
}

@article{cash,
  title={Parameter estimation in astronomy through application of the likelihood ratio},
  author={Cash, W.},
  journal={Astrophysical Journal},
  volume={228},
  pages={939--947},
  year={1979}
}

@article{pyke83,
  title={A uniform central limit theorem for partial-sum processes indexed by sets},
  author={Pyke, R.},
  journal={London Mathematical Society Lecture Notes Series},
  volume={79},
  pages={219--240},
  year={1983}
}

@article{pyke86,
  title={A uniform central limit theorem for set-indexed partial-sum processes with finite variance},
  author={Kenneth, A.S. and Pyke, R.},
  journal={The Annals of Probability},
  volume={14},
  number={2},
  pages={582--597},
  year={1986},
  publisher={Institute of Mathematical Statistics}
}

@article{cressie,
  title={Pearson's {$X^2$} and the loglikelihood ratio statistic {$G^2$}: A comparative review},
  author={Cressie, N. and Read, T.R.C.},
  journal={International Statistical Review},
 volume={57},
  number={1},
  pages={19--43},
  year={1989},
  publisher={JSTOR}
}

@article{algeri22,
  title = {K-2 rotated goodness-of-fit for multivariate data},
  author = {Algeri, S.},
  journal = {Physical Review D},
  volume = {105},
  issue = {3},
  pages = {035030},
  numpages = {10},
  year = {2022} 
}

@article{khm84,
  title={Martingale limit theorems for divisible statistics},
  author={Khmaladze, E.V.},
  journal={Theory of Probability and Its Applications},
  volume={28},
  number={3},
  pages={530--548},
  year={1984},
  publisher={SIAM}
}

@article{pearson1900,
  title={On the criterion that a given system of deviations from the probable in the case of a correlated system of variables is such that it can be reasonably supposed to have arisen from random sampling},
  author={Pearson, K.},
  journal={The London, Edinburgh, and Dublin Philosophical Magazine and Journal of Science},
  volume={50},
  number={302},
  pages={157--175},
  year={1900},
  publisher={Taylor \& Francis}
}


\end{document}